\documentclass[11pt]{amsart}
\usepackage[margin=1in]{geometry}
\usepackage{amsfonts}
\usepackage{amssymb}
\usepackage[dvips]{graphics}
\usepackage{epsfig}
\usepackage{euscript}
\usepackage{color}

\usepackage{fancyvrb}

\usepackage[pdftex]{hyperref}

 \newtheorem{thm}{Theorem}[section]
 
 \newtheorem{lemma}[thm]{Lemma}
 \newtheorem{prop}[thm]{Proposition}
 \newtheorem{conj}[thm]{Conjecture}
 \theoremstyle{definition}
 
 \theoremstyle{remark}

 \numberwithin{equation}{section}
 
 \def\idtyty{{\mathchoice {\mathrm{1\mskip-4mu l}} {\mathrm{1\mskip-4mu l}} %
{\mathrm{1\mskip-4.5mu l}} {\mathrm{1\mskip-5mu l}}}}

\newcommand{\supp}{\operatorname{supp}}

\newcommand{\caA}{{\mathcal A}}
\newcommand{\caB}{{\mathcal B}}
\newcommand{\caC}{{\mathcal C}}
\newcommand{\caD}{{\mathcal D}}
\newcommand{\caE}{{\mathcal E}}

\newcommand{\caG}{{\mathcal G}}
\newcommand{\caH}{{\mathcal H}}

\newcommand{\caK}{{\mathcal K}}

\newcommand{\caM}{{\mathcal M}}

\newcommand{\caO}{{\mathcal O}}
\newcommand{\caP}{{\mathcal P}}

\newcommand{\caR}{{\mathcal R}}
\newcommand{\caS}{{\mathcal S}}

\newcommand{\caU}{{\mathcal U}}
\newcommand{\caV}{{\mathcal V}}

\newcommand{\bbC}{{\mathbb C}}

\newcommand{\bbE}{{\mathbb E}}

\newcommand{\bbN}{{\mathbb N}}

\newcommand{\bbR}{{\mathbb R}}
\newcommand{\bbS}{{\mathbb S}}
\newcommand{\bbT}{{\mathbb T}}

\newcommand{\bbZ}{{\mathbb Z}}

\newcommand{\ie}{{\it i.e.\/} }

\newcommand{\iu}{\mathrm{i}}

\newcommand{\str}{^{*}}
\newcommand{\ep}[1]{\mathrm{e}^{#1}}

\newcommand{\Tr}{\mathrm{Tr}}

\newcommand{\abs}[1]{\left\vert #1 \right\vert}
\newcommand{\braket}[2]{\left\langle #1 , #2\right\rangle}

\newcommand{\ket}[1]{\left\vert #1\right\rangle}
\newcommand{\bra}[1]{\left\langle #1\right\vert}

\newcommand{\mean}[2]{\left\langle #1\right\rangle_{#2}}

\newcommand{\be}{\begin{equation}}
\newcommand{\ee}{\end{equation}}
\newcommand{\bea}{\begin{eqnarray}}
\newcommand{\eea}{\end{eqnarray}}
\newcommand{\beann}{\begin{eqnarray*}}
\newcommand{\eeann}{\end{eqnarray*}}

\begin{document}

\renewcommand{\thefootnote}{\fnsymbol{footnote}}
\title[Gapped phases]{Product vacua with boundary states and the classification of gapped phases}

\author[S. Bachmann]{Sven Bachmann}
\address{Department of Mathematics\\
University of California, Davis\\
Davis, CA 95616, USA}
\email{svenbac@math.ucdavis.edu}

\author[B. Nachtergaele]{Bruno Nachtergaele}
\address{Department of Mathematics\\
University of California, Davis\\
Davis, CA 95616, USA}
\email{bxn@math.ucdavis.edu}

\begin{abstract}
We address the question of the classification of gapped ground states in one dimension that cannot be distinguished by a local order parameter. We introduce a family of quantum spin systems on the one-dimensional chain that have a unique gapped ground state in the thermodynamic limit that is a simple product state but which on the left and right half-infinite chains, have additional zero energy edge states. The models, which we call Product Vacua with Boundary States (PVBS), form phases that depend only on two integers corresponding to the number of edge states at each boundary. They can serve as representatives of equivalence classes of such gapped ground states phases and we show how the AKLT model and its $SO(2J+1)$-invariant generalizations fit into this classification. 
\end{abstract}

\maketitle

\date{\today }

\footnotetext[1]{Copyright \copyright\ 2012, 2013 by the authors. This
paper may be reproduced, in its entirety, for non-commercial
purposes.}

%%%%%%%%%%%%%%%%%%%%%%%%%%%%%%%%%%%%%%%%%%%%%%%%%%%%%%%%%%%%%%%%%%%%%

\section{Introduction and statement of the results}

In this paper we apply the notion of automorphic equivalence introduced in \cite{BMNS}
to a class of quantum spin chains with nearest neighbor interactions. The main result of 
\cite{BMNS} establishes a relation between the ground states of two short range quantum 
spin Hamiltonians of which the interactions are the end points of a smooth curve and such 
that there is a uniform positive lower bound for the spectral gap of the models along the curve.
This relation is expressed by the existence of a cocycle of automorphisms of the algebra of 
quasi-local  observables that maps the set of ground states at one end point of the curve into
the set of ground states at the other end point. Moreover, these automorphism satisfy a 
Lieb-Robinson bound \cite{LR}, which expresses a strong quasi-local property, meaning that 
up to a small correction the support of an observable grows only a bounded amount under
the application of the automorphism. The physical interpretation of this relation between
the two sets of ground states is that they represent the same zero-temperature phase; one
can reach one from the other in a finite time without going through a (quantum) phase transition.
The main tools in \cite{BMNS} are Lieb-Robinson bounds and the so-called {\em quasi-adiabatic
continuation} technique initiated by Hastings \cite{LSM_H}, and further developed in 
\cite{LSM_NS}, \cite{BravyiHastings}, and \cite{BMNS}, where it is referred to as the 
{\em spectral flow}.

What we do here is constructing such a curve of Hamiltonians connecting the AKLT model
\cite{AKLT} and related models \cite{PhysRevB.78.094404} to models with a product ground 
state (the PVBS models introduced in \cite{PhysRevB.86.035149}). The tricky part is to prove 
that the gap has a positive lower bound all along the 
curve and uniformly in the length of the chain. The standard arguments of \cite{FCS} and 
\cite{VBSHam} do not yield a sufficiently good bound near one of the end points of the curve. 
This is related to the subtle difference between so-called {\em parent} versus {\em uncle} 
Hamiltonians (see Example 2 in  \cite{VBSHam} and the general theory in \cite{Uncle}). 
One of the new contributions in this work is the extraction of a uniform lower bound from the 
so-called Martingale Method. Its standard application to finitely correlated states obviously 
applies pointwise, but fails to provide a uniform bound whenever the ground state becomes 
product.

We view the construction of explicit curves of gapped models connecting the highly symmetric 
AKLT-type models with the PVBS toy models as a step toward a qualitative theory of gapped 
ground state phases of quantum many-body systems, with the ultimate goal of understanding 
these phases and the transitions between them as well as we do phase transitions at finite 
temperature. One significant difference is the occurrence of so-called
topological phases, \ie, the situation where different ground states coexists that cannot be 
distinguished by a local observable (order parameter). This situation is in stark contrast to the 
classical case of spontaneous symmetry breaking for which an order parameter distinguishing the 
different pure phases is readily identified.

The topological degeneracy of the ground states in models where it occurs, is usually best 
understood by considering the models defined on a set of lattices, rather than just one lattice.
Members of this set of lattices can have boundaries or non-trivial topology. The discovery 
about a decade ago \cite{Kitaev20032} of two-dimensional systems with topological 
order sparked great interest, specifically in view of possible applications in quantum information 
processing. Since then, however, this interest has broadened to aim more generally at
understanding the structure of gapped ground state phases and the quantum phase transitions
between them. 

The very criterion that determines whether two models should be considered to be in the same 
ground state phase has only recently received careful attention, both from a physical point of view 
\cite{HastingsWen, ChenGuWen} and from a more mathematical one \cite{Schuch, Hastings_Houches}. For 
systems that have a spectral gap in the thermodynamic limit, we have proposed in~\cite{BMNS} the 
notion of local automorphic equivalence to define a phase and shown how it relates to other criteria. 
Simply put, two models are in the same gapped ground state phase if there exists a quasi-local 
automorphism of the algebra of observables mapping the ground state space of one of the system to 
that of the other. Here, quasi-local is meant in the strong sense that the automorphism satisfies a 
Lieb-Robinson bound with error estimates that vanish faster than any power law.

The notion of automorphic equivalence becomes distinct from unitary equivalence, and more 
useful, when one applies it to infinitely extended systems. To study topological order one then
also considers Hamiltonians not only on one lattice, but on a family of lattices with boundaries 
or non-trivial topologies and the automorphisms map sets of  ground states onto each other. 
This  implies in particular that for any given lattice the ground state degeneracies are equal 
for models that are in the same phase.

This work is concerned with one-dimensional spin chains. The different types one-dimensional
lattices we consider are (i) finite chains, (ii) left and right  half-infinite chains, (iii) and the full 
bi-infinite chain. Despite the absence of interesting topology, this one dimensional situation 
allows us to emphasize the role played by boundaries and edge ground states. This yields a 
physically relevant refinement of the consensus that all gapped one-dimensional models with a 
unique ground state in the thermodynamic limit are in the same phase~\cite{Schuch}, usually referred to as the 
Haldane phase. We will show that models that share the same and simplest state on the doubly 
infinite chain can differ by their boundary behavior on the two possible half-infinite chains. These 
examples support the view that such systems can be classified by their boundary behavior.

To make this concrete, let us consider Hamiltonians $H_\Lambda$ defined on finite subsets of 
$\bbZ$ by
\begin{equation*}
H_\Lambda = \sum_{X\subset \Lambda}\Phi(X)\,,
\end{equation*}
where, for each finite subset of $\bbZ$, $\Phi(X)$ is a self-adjoint element of the algebra of 
observables belonging to $X$, which will be denoted by $\caA_X$. The dynamics and ground 
states of the model on the infinite 
lattice $\bbZ$, and the half-infinite lattices $(-\infty,0]$ and $[1,+\infty)$, are defined on the 
algebras of quasi-local observables $\caA_\bbZ$, $\caA_{(-\infty,0]}$, and $\caA_{[1,+\infty)}$,
respectively.

Let $\{T_x\}_x$ be the group of translations on $\bbZ$ and $\{\tau_x \}_x$ be the representation of the 
translations as automorphisms of the quasi-local algebra of observables $\caA_\bbZ$. We assume
that $\Phi$ has finite range and is translation invariant, $\Phi(T_x(X))=\tau_{x}(\Phi(X))$, for all 
$X\subset\bbZ$. For the infinite systems on the lattices $\Gamma=\bbZ, (-\infty,0], [1,+\infty)$, a ground 
state is a state of the quasi-local algebra $\caA_\Gamma$ such that
\begin{equation*}
\omega(X\str[H,X])\geq 0
\end{equation*}
for all local observables $X\in\caA_\Gamma$. Let $\caS^\Gamma$ be the set of ground states on 
$\Gamma$. All states obtained as thermodynamic limits of finite-volume ground states are ground 
states of the infinite system in this sense. In order not to unnecessarily complicate things,
we will say that the system is gapped if there exists $\gamma>0$ such that the difference 
between the smallest and the next-smallest eigenvalue of $H_\Lambda$, is bounded below by 
$\gamma$,  for all finite intervals $\Lambda$. We will denote by $\caE_d$, the set of states of a 
quantum system with a $d$-dimensional
Hilbert space. E.g., $\caE_2$ is the standard Bloch sphere describe all possible states of a 
qubit, and $\caE_1$ contains a single pure state. For two convex sets of states $A$ and $B$, 
$A\cong B$ will mean isomorphism of convex sets, and for two sets of states $\caS_0$ and $
\caS_1$ of a given algebra of observables $\caA$, $\caS_0\sim \caS_1$ will mean
\begin{equation}
\caS_1 =\{\omega\circ\alpha \mid \omega\in\caS_0\},
\label{homeo}\end{equation}
with $\alpha$ a quasi-local automorphism of $\caA$, satisfying 
$$
\Vert [\alpha(A),B]\Vert \leq \Vert A\Vert \,\Vert B\Vert F(d(\supp(A), \supp(B)),
$$
for a function $F$ that decays faster than any power law, $d$ is the standard distance between
subsets if the lattices and $\supp(C)$ denotes the smallest set $X$, such that $C\in\caA_X$.
\begin{conj}\label{conj:Equivalence}
Let $\Phi_0$ and $\Phi_1$ be two translation invariant finite range interactions. Assume that for 
$i=0,1$, there exist integers $l_i$ and $r_i$ such that
\begin{equation*}
\caS^{\bbZ}_i\cong\caE_1\,,\qquad 
\caS^{(-\infty,0]}_i\cong\caE_{r_i}\,,\qquad 
\caS^{[1,+\infty)}_i\cong\caE_{l_i}\,,
\end{equation*}
Then, the following are equivalent:
\begin{enumerate}
\item $r_0=r_1, l_0=l_1$;
\item There exists a continuously differentiable family of translation invariant finite range interactions 
$\Phi(s)$, with $s\in[0,1]$, such that $\Phi(0) = \Phi_0$ and $\Phi(1) = \Phi_1$, and there is a uniform 
lower bound $\gamma >0$ for the gap of all models defined by $\Phi(s)$, $s\in [0,1]$.
\item $\caS_0^\Gamma \sim \caS_1^\Gamma$ for $\Gamma=\bbZ, (-\infty,0]$, and $[1,+\infty)$
\end{enumerate}
\end{conj}

Clearly, (iii) implies (i), since the automorphism $\alpha$ in (\ref{homeo}) induces 
a homeomorphism between two sets of states of the form $\caE_k$, regarded as
topological manifolds. That (ii) implies (iii) is proved constructively in~\cite{BMNS}. 
The remaining implication is conjectural. The condition that the sets of ground states
$\caS^\Gamma$ are isomorphic to $\caE_d$ for some positive integer $d$, is equivalent
to assuming that they are the state spaces of the full matrix algebra $M_d$ of complex
$d\times d$ matrices. It seems plausible that the assumption that $\caS^{\bbZ}$ is
a single point, which means a unique ground state in the bulk, implies this condition.
In this paper, we limit ourselves to this situation, but it is straightforward to generalize the 
conjecture to cases where the sets $\caS^\Gamma$ are isomorphic to the state spaces 
of a direct sum of matrix algebras. E.g., if the bulk ground states is two-fold degenerate,
\ie, two pure ground states in the bulk, $\caS^{\bbZ}$  is isomorphic to the segment $[0,1]$,
which is the state space of $\bbC \oplus\bbC$. We then expect that $\caS^{(-\infty,0]}$ and
$\caS^{[1,+\infty)}$, if finite-dimensional, will be isomorphic to the state spaces of the direct 
sum of two square matrix algebras. This is the situation needed to describe the standard
(non-topological) ground state phase transitions.

The duality between `topological order' in the bulk and the existence of non trivial edge states
is expected to extend to two-dimensional systems. See for 
example~\cite{PhysRevB.83.245134,Schuch2010} for a detailed study of a class of examples. 
In fact, the existence of gapless edge modes in a system that is gapped in the bulk was 
originally proposed as a way to measure topological order in fractional quantum Hall 
systems~\cite{We95}. Further, the equality of bulk and edge conductances in the two-dimensional
quantum Hall effect~\cite{ElGr02}, or the bulk-edge correspondence for topological 
insulators~\cite{GrPo} can be seen as instances of the same principle.

We start this paper by presenting a detailed study of the so-called PVBS models we introduced in  
\cite{PhysRevB.86.035149}. These are a continuous family of translation invariant Hamiltonians 
with nearest neighbor interaction and finitely correlated ground states. The spectral gap does not 
close in the thermodynamic limit except when specified parameters take on a critical value (which
we have normalized to be 1). On the infinite chain, they share a unique product ground state. 
However, they have degenerate ground states on half-infinite chains which have a simple 
interpretation. There are $n$ distinguishable particles that can be added to the vacuum without 
raising the energy if no other particle of the same type is already present and if they can bind to a 
boundary. Among the $n$ types of particles,  $n_L$ can bind to a left edge while $n_R=n-n_L$ bind 
to a right edge, yielding a degeneracy of $2^n$ on finite chains, and $2^{n_L}$, resp $2^{n_R}$, on 
the two possible half-infinite chains. We call these models Product Vacua with Boundary States. We 
illustrate the classification of phases by proving that two PVBS models belong to the same phase if 
and only if they share the same number of left- and right- binding particles. The simplicity of the 
models make them good candidates as representatives of the gapped phases and  it follows from 
our results that Conjecture \ref{conj:Equivalence} holds within this restricted class of models.

Next, we show that the well-known spin-$1$ antiferromagnetic chain known as the AKLT 
model~\cite{AKLT} fits in the PVBS classification, namely that it belongs to the PVBS phase with 
one of each of the two types of particles ($n_L=n_R=1$). These proofs rely on two ingredients: 
First that a local automorphism can be 
constructed from a smooth path of gapped Hamiltonians interpolating between the two models in 
question. And second on a criterion for models with matrix product ground states (MPS) to have a 
spectral gap~\cite{VBSHam}. Despite the fact that much can be said by simply considering the 
algebraic deformation of the matrices generating the ground state spaces as in~\cite{Schuch}, we 
shall emphasize that this may not be sufficient. First of all, the same spaces can be ground state 
spaces of both gapped and gapless frustration-free Hamiltonians, see~\cite{Uncle}. Secondly, the 
pointwise construction of Hamiltonians from a smooth family of matrix product maps is not a 
continuous operation. More fundamentally, smoothly connecting MPS described by matrices of 
different dimensions requires the use of non minimal representations in which the spectral gap is in 
general not under good control. The central objects of the smooth deformation must be the ground 
state spaces and the Hamiltonians themselves.

Already from the few examples that we study here in detail, it becomes clear that the classification of 
gapped phases that we obtain is finer than those previously pursued in the literature (see, e.g.,
\cite{PhysRevB.83.035107,PhysRevB.84.235128}.) By this we do not mean that there are 
more distinct gapped phases in the bulk than previously claimed, but that our finer analysis provides 
more information about the ground state phase diagrams. To illustrate this, consider two 
models defined by finite-range interactions that have a unique ground state in the bulk and a
non-vanishing spectral gap above them. If the half-infinite chains with these interactions 
have inequivalent ground state spaces in the sense of above, any smooth curve of finite-range
interactions connecting the two given models will contain a point where the spectral gap closes.
This does not contradict the well-established fact that each model can be smoothly connected
to a model with a simple product ground state. The interactions of the models with the unique 
product ground state will however not be the same and will, in particular, have different edge state
behavior.

One possible family of models generalizing the AKLT model carries the fundamental representation 
of $SO(d)$ and its nearest-neighbor interaction is $SO(d)$-invariant~\cite{PhysRevB.78.094404}. In 
the last section of this article, we propose a smooth path of Hamiltonians between these $SO(d)$ 
models for odd $d=2J+1$ and the PVBS models with $J$ particles bound to each edge. Here again, 
the starting point is an algebraic deformation, namely of the Clifford algebra underlying the $SO(d)$ 
invariance of the models.

Other generalizations in one dimension include PVBS models allowing more than one particle of 
each type, generating arbitrary degeneracies at each edge with a unique bulk state, or multiple bulk 
states, namely systems with a broken discrete symmetry. Local continuous symmetries can also be 
included by identifying a subset of particles. Of course, adding symmetries enriches the classification 
of gapped phases and the possibilities of transition between them~\cite{PhysRevB.83.035107, 
PhysRevB.84.235128}.

The class of PVBS models is rich enough to cover a large range of behaviors of one-dimensional 
quantum spin systems. Moreover, their simplicity allows for a complete understanding of their ground 
state properties, and they appear amenable for a general study of the transitions between the ground 
state phases they represent: A phase transition arises in a PVBS model whenever one of the edge 
particles unbinds from one boundary and binds to the other. At the critical point the particle becomes 
a massless excitation of the vacuum and the gap closes.

%%%%%%%%%%%%%%%%%%%%%%%%%%%%%%%%%%%%%%%%%%%%%%%%%%%%%%%%%%%%%%%%%%%%%

\section{Product vacua with boundary states} \label{Sec:PVBS}

%%%%%%%%%%%%%%%%%%%%%%%%%%%%%%%%%%%%%%%%%%%%%%%%%%%%%%%%%%%%%%%%%%%%%

We consider a quantum spin system on the chain $\bbZ$. Any point $x\in\bbZ$ carries a quantum system with Hilbert space $\bbC^d$, for some $d\in\bbN$, equipped with an orthogonal basis $\{e_0,\ldots,e_n\}$, where $n=d-1$. We shall denote the finite subset $[a,b]\cap\bbZ$ simply by $[a,b]$, for $a<b\in\bbZ$. The Hilbert space $\caH_N$ of a finite chain of length $N= b-a+1$ is given by
\begin{equation*}
\caH_N = \otimes_{x=a}^b\bbC^d\,.
\end{equation*}
We first define a nearest-neighbor interaction $h \in \caH_2$ as follows. For $0\leq i<j\leq n$, let
\begin{align}
\phi^{ij}&= e_i\otimes  e_j-\ep{\iu\theta_{ij}}\lambda_i^{-1} \lambda_j \, e_j\otimes  e_i\,, \label{phiij}\\
\phi^{ii}&= e_i\otimes  e_i \label{phiii}\,,
\end{align}
where $\theta_{ij}\in\mathbb{R}$, $\theta_{ij}=-\theta_{ji}$, and $0\neq \lambda_i\in \mathbb{C}$. By redefining the phases $\theta_{ij}$, we can assume $\lambda_i >0$ without loss of generality. The interaction $h$ is defined as the orthogonal projection on the space spanned by these vectors. 

Given a fixed set of parameters $\Lambda=\{\lambda_i:0\leq i\leq n\}$ and phases $\Theta=\{\theta_{ij}:0\leq i<j\leq n\}$, the PVBS Hamiltonian on $[a,b]$ is given by
\begin{equation*}
H_{[a,b]} = \sum_{x=a}^{b-1}h_{x,x+1}
\end{equation*}
where $h_{x,x+1}$ is a copy of $h$ acting on the neighboring sites $x$ and $x+1$. We consider four different realizations of the PVBS models, namely when defined on a finite chain $[a,b]$, on either of the two half infinite chains with one edge, and on the infinite chain $\bbZ$. Of course, the various infinite limits involved here must be taken carefully.

Let $\gamma(\Lambda, \Theta, N)$ be the spectral gap above the ground state energy of the PVBS model on a chain of length $N$, with parameters $\Lambda$ and $\Theta$. In the sequel, we shall set $\lambda_0=1$. The special role played by $\lambda_0$ will be discussed later.
\begin{prop}[Gap bounds] \label{Prop:PVBS_Gap}
The spectral gap for the Hamiltonian $H_{[a,b]}$ is bounded above by
\begin{equation}\label{MassGap}
\gamma(\Lambda, \Theta, N) \leq \min\left\{\left(1-\frac{2}{\lambda_i^{-1} + \lambda_i}\right)\left(1 + \frac{C_i-1}{N-C_i}\right):1\leq i \leq n\right\}\,,
\end{equation}
for any $N=b-a+1$, where $C_i = (1+\lambda_i)/\abs{1-\lambda_i}$. Moreover, if $\lambda_i\neq 1$ for all $i\neq0$, then
\begin{equation*}
\gamma:=\inf_{N\in\bbZ}\gamma(\Lambda, \Theta, N)>0 \,.
\end{equation*}
\end{prop}
By standard results~\cite{FCS}, the second statement implies a spectral gap for the Hamiltonian in the GNS representation 
of the quasi-local algebra induced by the ground states in the thermodynamic limit. The notion of a ground state phase of 
quantum spin systems was discussed in~\cite{BMNS}, where automorphic equivalence between ground state spaces was 
introduced. Given a differentiable curve of local Hamiltonians $H(s)$, for $s\in[0,1]$, such that there is a lower bound for the 
spectral gap which is uniform in the size of the system and in the parameter $s$, then there exists a quasi-local spectral flow 
$\alpha_{s,s'}$ of the algebra of observables mapping the ground state space of $H(s')$ to that of $H(s)$. For finite chains it 
is given by unitary conjugation, and can be extended to infinite systems using its good locality properties. 

In the present case, let $\Lambda_L := \{i:\lambda_i<1\}$, $\Lambda_R := \{i:\lambda_i>1\}$, and
\begin{equation*}
n_L:=\abs{\Lambda_L}\,,\qquad n_R:=\abs{\Lambda_R}\,.
\end{equation*}
Note that the results about the ground states of the PVBS models obtained in \cite{PhysRevB.86.035149} (also described 
in detail below) imply that the PVBS models have all the properties of the preamble of Conjecture~\ref{conj:Equivalence}. 
In this section, we say that two quantum spin systems are equivalent if they satisfy {\em all} three conditions in Conjecture~\ref{conj:Equivalence}. The results proved in this section imply, among other things, that these three conditions are indeed 
equivalent as conjectured.
\begin{thm}[PVBS classes] \label{thm:PVBS_phases}
Two PVBS models with respective $n_{L,R}^{1,2}$ are equivalent if and only if
\begin{equation*}
n_L^1 = n_L^2\qquad\text{and}\qquad n_R^1 = n_R^2\,.
\end{equation*}
\end{thm}
Before we enter into the proof of this theorem, let us briefly discuss its importance. The family of PVBS models is parametrized 
by the sets $\Lambda$ and $\Theta$. Theorem~\ref{thm:PVBS_phases} distinguishes among them a countable set of equivalence classes characterized by a simple and explicit criterion. In fact, all these translation invariant models have the same, unique, gapped ground state in the thermodynamic limit. However, they do not belong to the same ground state phase. As we will discuss in more detail later, their properties differ by their boundary behavior which is precisely classified by the two integers $n_L$ and $n_R$.

Automorphic equivalence of the models -- or equivalently their belonging to the same ground state phase -- is obtained by constructing a smooth family of gapped Hamiltonians interpolating between the two models under consideration, see~\cite{BMNS}. Reciprocally, the theorem implies that along any curve relating two models belonging to two phases indexed by different $(n_L, n_R)$, the gap must close and there is a quantum phase transition. 

The rest of the section is organized as follows. First, we construct the ground states of the PVBS models on finite chains. These models are frustration free and finitely correlated, so that it will be convenient to use their representation as matrix product states (MPS). The first part of Proposition~\ref{Prop:PVBS_Gap} follows by a variational argument. Its second statement uses the martingale method of~\cite{VBSHam} and a crucial geometric estimate associated with ground state projections. We construct gapped paths of Hamiltonians within the set of PVBS models and Theorem~\ref{thm:PVBS_phases} follows from the gap estimates of Proposition~\ref{Prop:PVBS_Gap}.

\subsection{Matrix product states (MPS)}\label{sub:MPS}

We briefly review the basic construction and properties of purely generated matrix product states, also known as finitely correlated states, see~\cite{FCS}. Let $e\in\caM_k$ be a positive element, $\rho\in\caM_k\str$ be a state on $\caM_k$, and let $\bbE:\caM_d\otimes\caM_k\rightarrow\caM_k$ be a completely positive map such that
\begin{equation*}
\bbE(\idtyty\otimes e) = e\,,\qquad \rho (\bbE(\idtyty\otimes B) ) = \rho(B)\,,
\end{equation*}
for all $B\in\caM_k$. The matrix product state $\omega$ on the chain algebra $\caA$ generated by $(\bbE,e,\rho)$ is the extension of
\begin{equation*}
\omega(A_a\otimes\cdots\otimes A_b) := \rho(e)^{-1}\rho\left(\bbE_{A_a}\circ\cdots\circ\bbE_{A_b}(e)\right)\,,
\end{equation*}
where $A_i\in\caM_d$ and we have denoted $\bbE_A(B) = \bbE(A\otimes B)$. If $\bbE$ is a pure map, then there exists $V:\bbC^k\rightarrow\bbC^d\otimes\bbC^k$ such that $\bbE(A\otimes B) = V\str (A\otimes B) V$.

Now, the breaking or not of translation symmetry can be expressed as a spectral property of the transfer operator $\widehat{\bbE}:=\bbE_\idtyty$. The state $\omega$ cannot be decomposed into periodic components iff $1$ is an isolated, non degenerate eigenvalue of $\widehat{\bbE}$ and all other eigenvalues have a strictly smaller modulus -- the peripheral spectrum is trivial. This implies the following exponential convergence
\begin{equation*}
\lim_{N\to\infty} \widehat{\bbE}^{(N)}(B) = \widehat{\bbE}^{(\infty)}(B) := \rho(B) e
\end{equation*}
and since
\begin{equation*}
\omega(A_a\otimes\idtyty\otimes\cdots\otimes\idtyty\otimes A_b) = \rho(e)^{-1}\rho\left(\bbE_{A_a}\circ\widehat{\bbE}^{(N-2)}\circ\bbE_{A_b}(e)\right)\,,
\end{equation*}
the exponential decay of correlations in the ground state. We shall say that the representation of $\omega$ is minimal if both $\rho$ and $e$ are invertible.

The state $\omega$ defined on the infinite chain is the thermodynamic limit of `valence bond vectors' defined as follows. The map $V$ can be described by the set $\{v_i\}_{i=1}^d$ of $k\times k$ matrices such that
\begin{equation*}
V\chi = \sum_{i=1}^d e_i\otimes v_i\chi\,.
\end{equation*}
The range $\caG_N$ of the map $\Gamma_N:\caM_k\rightarrow (\bbC^d)^{\otimes N}$ defined by
\begin{equation}
\Gamma_N(B) := \sum_{i_1,\ldots,i_N=1}^d \Tr (B v_{i_N}\cdots v_{i_1})\, e_{i_1}\otimes\cdots\otimes  e_{i_N}\,,
\label{mps}
\end{equation}
is the set of matrix product states on a finite chain of length $N$. If $\{\chi_p\}_{p=1}^k$ is an orthonormal basis of $\caM_k$, a simple calculation yields
\begin{align}\label{ScalarProdGamma}
\braket{\Gamma_N(B_L)}{A\Gamma_N(B_R)} &= \sum_{i_1\ldots i_N=0}^{n}\sum_{j_1\ldots j_N=0}^{n}\overline{\Tr(B_L v_{i_N}\cdots v_{i_1})}\Tr(B_R v_{j_N}\cdots v_{j_1}) \braket{e_{i_1}\cdots e_{i_N}}{Ae_{j_1}\cdots e_{j_N}}\nonumber \\
&= \sum_{p,q=1}^k\braket{\chi_p}{\bbE_A^{(N)}\left(B_L\str \ket{\chi_p}\bra{\chi_q}B_R\right)\chi_q}\,.
\end{align}

In a minimal representation, the spaces $\caG_N$ always satisfy a crucial intersection property. The following lemma presents a sufficient condition for the property to hold which does not require faithfulness of the state $\rho$.

We consider the case where the matrices of the MPS are generators of a quadratic algebra. Let $\caU:=\mathrm{span}\{g_i:1\leq i\leq d\}$ and $\oplus_{\alpha=0}^\infty \caU^{\otimes \alpha}$ the tensor algebra generated by $\caU$. Given a multiindex $I=(i_1,\ldots,i_r)$ where $1\leq i_j\leq d$, $g^I$ denotes the monomial $g^I = g_{i_1}\cdots g_{i_r}\in \caU^{\otimes r}$. We will also denote $\bar I =(i_r,\ldots,i_1)$. Let $\caR\subset\caU^{\otimes 2}$ be the set of quadratic relations. Choosing a basis of $\caU^{\otimes 2} / \caR$ is equivalent to choosing a set $S^{(2)}\subset[1,d]^2$ such that
\begin{equation*}
g_{i}g_j = \sum_{(k,l)\in S^{(2)}} m^{\phantom{{i,j}}k,l}_{i,j}g_k g_l\qquad (i,j)\in \overline{S^{(2)}}\,,
\end{equation*}
and $m^{\phantom{{i,j}}k,l}_{i,j} = \delta^{\phantom{{i,j}}k,l}_{i,j}$ whenever $(i,j)\in S^{(2)}$. Further, consider the multiindices
\begin{equation*}
S^{(r)}:=\left\{(i_1,\ldots,i_r):(i_j,i_{j+1})\in S^{(2)}, 1\leq j < r\right\}\,.
\end{equation*}
Finally, let $\pi$ be a representation of the algebra on $\bbC^k$, and we denote $v_i = \pi(g_i)$ for all $i$.
\begin{lemma}\label{lma:IntProp}
Assume that for all $r\geq 3$, the set $\{v^I:  I\in S^{(r)}\}$ is linearly independent in $\caM_k$. Then the following intersection property holds:
\begin{equation} \label{Intersection}
\mathcal{G}_{l}=\bigcap_{x=0}^{l-2} (\mathbb{C}^d)^{\otimes x}\otimes
\mathcal{G}_{2}\otimes (\mathbb{C}^d)^{\otimes (l-2-x)}
\end{equation}
for all $l\geq 2$.
\end{lemma}
\begin{proof}
By induction, it suffices to prove
\begin{equation*}
\caG_{r+1} = \caG_r\otimes \bbC^d \cap \bbC^d\otimes \caG_r =:\tilde\caG_r
\end{equation*}
for $r\geq 2$.

\textsc{Part I.} $\caG_{r+1}\subseteq \tilde\caG_r$.

A vector $\phi$ belongs to $\caG_{r+1}$ if there is a matrix $B$ for which $\phi_{\bar I} = \Tr(Bv^I)$ for $I\in[1,d]^{r+1}$ and it is in $\tilde\caG_r$ if there are $(R_k)_{k=0}^n$ and  $(L_i)_{i=0}^n$ with $\phi_{\bar I} = \Tr(R_{i_1}v_{i_2}\cdots v_{i_{r+1}}) = \Tr(v_{i_1}\cdots v_{i_{r}} L_{i_{r+1}})$. The inclusion $\subseteq$ follows by choosing $R_{i_1} = Bv_{i_1}$ and $L_{i_{r+1}} = v_{i_{r+1}}B$.

\textsc{Part II.} $\caG_{r+1}\supseteq \tilde\caG_r$. 

First, we note that the linear independence condition implies that $\{g^I: I\in S^{(r)}\}$ is linearly independent at the level of the algebra. Indeed, if it were not the case, there would exist a non-trivial relation $\sum_I \alpha_I g^I = 0$. In the representation, this implies that $0 = \pi(\sum_I \alpha_I g^I) = \sum_I \alpha_I \pi(g^I) = \sum_I \alpha_I v^I$ which is a contradiction. Furthermore, $\mathrm{span}\{g^I: I\in S^{(r)}\} = \caV^{\otimes r}$, as every monomial $g^I$ with $(i_l,i_{l+1})\in \overline{S^{(2)}}$ can be expressed as a linear combination of monomials with $(j_l,j_{l+1})\in S^{(2)}$. This, and the linear independence, imply the existence of a unique matrix $M^I_{\phantom{I}J}$, of size $d^{r+1}\times\abs{S^{(r+1)}}$, such that
\begin{equation*}
g^I = \sum_{J\in S^{(r+1)}} M^I_{\phantom{I}J}g^J\,,
\end{equation*}
for all $I$, where $M^I_{\phantom{I}J} = \delta^I_{\phantom{I}J}$ whenever $I\in S^{(r+1)}$.

Now, $\phi\in\caG_{r+1}^\perp$ iff $\sum_{I\in[1,d]^{r+1}} \overline{\phi_{\bar I}}\Tr(Bv^I) = 0$ for all $B$, \ie
\begin{equation*}
\phi\in\caG_{r+1}^\perp \quad \Longleftrightarrow \quad \sum_{I\in[1,d]^{r+1}} \overline{\phi_{\bar I}} v^I = 0\,.
\end{equation*}
On the other hand, if $\phi\in\tilde\caG_r$, then also
\begin{equation} \label{phiRelabel}
\phi_{\bar{I}} = \sum_{J\in S^{(r+1)}} M^I_{\phantom{I}J} \phi_{\bar{J}}\,.
\end{equation}
Hence,
\begin{equation*}
\sum_{I\in[1,d]^{r+1}}\overline{\phi_{\bar{I}}}v^I 
=\sum_{I\in[1,d]^{r+1}} \sum_{K\in S^{(r+1)}} \overline{M^I_{\phantom{I}K}} \overline{\phi_{\bar{K}}} \sum_{J\in S^{(r+1)}} M^I_{\phantom{I}J} v^J
=\sum_{J\in S_R^{(r+1)}} \overline{(M\str M \phi)_{\bar J}}\, v^J\,,
\end{equation*}
so that $\sum_{I}\overline{\phi_{\bar{I}}}v^I = 0$ implies $(M\str M) \phi = 0$, by the assumption. Since the matrix $M$ has the following block structure:
\begin{equation*}
M = 
 \begin{pmatrix}
  \idtyty \\ N
 \end{pmatrix}\,,
\end{equation*}
it follows that $M\str M = \idtyty + N\str N \geq \idtyty$, so that $\det M\str M \neq 0$. In particular, $(M\str M) \phi = 0$ implies $\phi_{\bar{J}}=0$ for $J\in S^{(r+1)}$, and by~(\ref{phiRelabel}) $\phi=0$. Hence, $\caG_{r+1}^\perp \cap \tilde\caG_r = \{0\}$. This, combined with $\caG_{r+1} \subseteq \tilde\caG_r$, implies that $\caG_{r+1} = \tilde\caG_r$.
\end{proof}

The last general result we recall here is about a uniform lower bound for the spectral gap, see~\cite{VBSHam}. Assume that $H_{[1,N]}\geq 0$ is a positive Hamiltonian corresponding to a translation invariant nearest neighbor interaction and such that $\mathrm{Ker}(H_{[1,N]}) = \caG_{[1,N]}$ for all $N\geq 2$. Let $G_N$ be the orthogonal projection onto $\caG_N$. If there exists $k$, an $\epsilon_k < 1/\sqrt{k}$ and $N = N(k)$ sufficiently large such that
\begin{equation} \label{martingale}
g_{k,N} := \left\Vert G_{[N-k+2, N+1]}\left(G_{[1, N]} - G_{[1, N+1]}\right)\right\Vert \leq \epsilon_k\,.
\end{equation}
Then the spectral gap $\gamma_N$ of $H_{[1,N]}$ is bounded below,
\begin{equation}\label{GapMartingale}
\gamma_N \geq \frac{\gamma_k}{k-1}\left(1-\epsilon_k \sqrt{k}\right)^2\,,
\end{equation}
uniformly on the length of the chain, where $\gamma_k$ is the spectral gap of the Hamiltonian on the chain of length $k$.

\subsection{The MPS representation of the PVBS models and gap estimates}

We now turn to the analysis of the PVBS systems using their matrix product representation. In this case, $d=n+1$. The quadratic relations satisfied by the generating matrices $\{v_i\}_{i=0}^n$ are given by
\begin{align}
v_i v_j&=\ep{\iu\theta_{ij}} \lambda_i \lambda_j^{-1} v_j v_i \,, \qquad i\neq j\,,\label{cr1}\\
v_i^2 &= 0\,, \qquad i\neq 0\,,\label{cr2}
\end{align}
and $\lambda_0 = 1$ corresponds to a choice of normalization. Recall that a vector $\phi\in\caG_2^\perp$ iff $\sum_{i,j=0}^n\overline{\phi_{ij}}v_jv_i = 0$, where $\phi_{ij} = \braket{e_i\otimes e_j}{\phi}$. Comparing~(\ref{cr1},\ref{cr2}) with~(\ref{phiij},\ref{phiii}), it is clear that the vectors $\phi^{ij}$ and $\phi^{ii}$ span $\caG_2^\perp$ if $v_iv_j\neq 0$ for all pairs $i\neq j$.
\begin{lemma}
There exists a $2^n$-dimensional representation of the commutation relations~(\ref{cr1},\ref{cr2}) with $v_iv_j\neq 0$ for $0\leq i<j\leq n$ and $v_0^2\neq0$.
\end{lemma}
\begin{proof}
In an arbitrary basis of $\bbC^2$, let
\begin{equation*}
\sigma^+ = \begin{pmatrix} 0 & 1 \\ 0 & 0 \end{pmatrix} \qquad
\sigma^- = \begin{pmatrix} 0 & 0 \\ 1 & 0 \end{pmatrix}\,,
\end{equation*}
and
\begin{equation*}
d_i = \begin{pmatrix} 1 & 0 \\ 0 & \lambda_i \end{pmatrix}  \qquad
P_{ij} = \begin{pmatrix} \ep{\iu\theta_{ij}/2} & 0 \\ 0 & 1 \end{pmatrix}\,.
\end{equation*}
The matrices given by
\begin{align}
v_0 &= \bigotimes_{i=1}^n P_{0i}^2d_i\,, \label{v0}\\
v_i &= \bigotimes_{j=1}^{i-1} P_{ij}d_j \otimes \sigma^+ \otimes \bigotimes_{k=i+1}^n P_{ik}d_k\,,\quad 1\leq i \leq n\,. \label{vi}
\end{align}
are of dimension $2^n$ and satisfy the commutation relations. Indeed~(\ref{cr1}) follows from $\sigma^+d_i = \lambda_i d_i \sigma^+$ and $P_{ij}\sigma^+ = \ep{\iu\theta_{ij}/2}\sigma^+P_{ij}$ as well as $[P_{jk},d_i] = 0$ for all $i,j,k$. Eq.~(\ref{cr2}) is a direct consequence of $(\sigma^+)^2 = 0$.
\end{proof}
\begin{lemma}
For $a,b\in\bbZ$, with $b-a+1 = N\geq2$, the ground state space of $H_{[a,b]}$ is equal to $\caG_N$.
\end{lemma}
\begin{proof}
The intersection property for the spaces $\caG_N$ follows from Lemma~\ref{lma:IntProp}, applied to the PVBS algebra~(\ref{cr1}, \ref{cr2}). In the notation thereof, we choose
\begin{equation*}
S^{(2)} = \{(0,0)\} \cup \{(k,l):0\leq k<l\leq n\}\,,
\end{equation*}
with $m^{i,i}_{\phantom{i,i}k,l} = 0$ for all $(k,l)\in S^{(2)}$ if $i\neq 0$, and $m^{i,j}_{\phantom{i,j}k,l} = \delta^{i,j}_{\phantom{i,j}j,i} \ep{\iu\theta_{ij}} \lambda_i \lambda_j^{-1}$. Further,
\begin{equation*}
S^{(3)} = \{(0,0,0)\} \cup \{(0,0,k):k=1,\ldots,n\}\cup \{(ijk): 0\leq i<j<k \leq n\}\,.
\end{equation*}
Note that for $r$ large enough, there is a bijection between $S^{(r)}$ and the set of subsets of $\{1,\ldots,n\}$, so that
\begin{equation*}
\mathrm{dim}\{v^J:J\in S^{(r)}\}\leq 2^n
\end{equation*}
for all $r$ and it is equal to $2^n$ for $r\geq n$. The linear independence of this set follows immediately from the tensor product structure of the representation and the independence of any $d_i$ with $\sigma^+$.

Now, by definition of the interaction and the intersection property, we have that
\begin{equation*}
\mathrm{Ker}(H_{[a,b]}) = \caG_N
\end{equation*}
for $N\geq 2$. $H_{[a,b]}$ being the sum of projections, it is a non negative operator so that $\caG_N$ is its ground state space.
\end{proof}
We emphasize the fact that the intersection property refers only to the spaces $\caG_N$ themselves. The matrix product representation of the vectors therein is only a convenient tool to describe them. As $\caG_N = \mathrm{Ran}(\Gamma_N)$, the ground state spaces can be explicitly described. It is easy to see that any product $v_{i_N}\cdots v_{i_1}$ in which an index $j\neq0$ appears more than once vanishes. Indeed, using the commutation relations, it is possible to bring the two $v_j$'s next to each other, and recall that $v_j^2 = 0$. Moreover, for any subset $S=\{s_1,\ldots,s_m\}\subseteq\{1,\ldots,n\}$, there exists a vector $\psi_{[a,b]}^S$ which is a linear combination of all basis vectors of the form
\begin{equation*}
e_0\otimes\cdots\otimes e_{s_1}\otimes\cdots\otimes e_{s_m} \otimes\cdots\otimes e_{0}\,.
\end{equation*}
where each $ e_{s_1},\ldots,  e_{s_m}$ appears exactly once. Explicitly,
\begin{equation} \label{GSvectors}
\psi_{[a,b]}^S = \Gamma_N(B^S)\,,
\end{equation}
where
\begin{equation}
B^S = P^{\otimes (s_1-1)}\otimes \sigma^-\otimes P^{\otimes (s_2-s_1-1)}\otimes \sigma^- \otimes \\
\cdots \otimes\sigma^-\otimes P^{\otimes (n-s_m)} \label{Bs}
\end{equation}
and $P = \sigma^+\sigma^-$. For example,
\begin{equation} \label{iGS}
\psi_{[a,b]}^{\{i\}} = \sum_{x=a}^b \left(\ep{\iu\theta_{i0}}\lambda_i\right)^{x-a+1}  e_0\otimes\cdots\otimes e_0\otimes e_i\otimes e_0\otimes\cdots\otimes e_0
\end{equation}
where $e_i$ is at site $x$ in each term of the sum. Note that $\psi_{[a,b]}^S$ is unique, given the subset $S$. Indeed, the commutation relations determine all components of the vector up to normalization. More generally, 
\begin{equation*}
\left\vert\braket{e_0\otimes\cdots\otimes e_{s_1}\otimes\cdots\otimes e_{s_m}\otimes \cdots\otimes e_0}{\psi_{[a,b]}^{S}}\right\vert^2 = \prod_{j=1}^m\lambda_{s_j}^{2 x_j}\,.
\end{equation*}
Finally, the vector corresponding to $m=0$ is the simple product
\begin{equation*}
\Omega =  e_0\otimes\cdots\otimes e_0\,.
\end{equation*}
As $S\neq S'\Rightarrow \langle\psi_{[a,b]}^S,\psi_{[a,b]}^{S'}\rangle=0$, the dimension of the span of such vectors is $\abs{\caP(\{1,\ldots,n\})} = 2^n$ for all $N\geq n$. Since this is also the dimension of $\{v^J:J\in S^{(N)}\}$, we have that
\begin{equation*}
\mathrm{span}\{\psi_N^{S}:S\in\caP(\{1,\ldots,n\})\} = \caG_N\,.
\end{equation*}

Since the initial space of $\Gamma_N$ has dimension $4^n$, its kernel is non trivial for all $N$. In the language of matrix product states, the fact that $\Gamma_N$ is not injective for all chain lengths implies that the representation is non minimal. And indeed, the minimal representation here is one dimensional, as the state on the infinite chain is a product. Here, the larger algebra allows for a richer structure of the state when boundaries are present.

The natural interpretation of the ground states in finite volume is that of a product vacuum $\Omega$ upon which $n$ types of particles can be added, but at most one of each type. Each particle is bound to one of the edges, to the left edge if $\lambda_i<0$ and to the right edge if $\lambda_i>0$, as can be read from the exponential decay of the wavefunction's amplitude as the particle is placed farther away from its boundary, see~(\ref{iGS}). 

Recall that $\widehat{\bbE}$ is the transfer operator on $\caM_k$ associated to the matrices $v_i$ generating the PVBS.
\begin{lemma}
For any choice $\{\lambda_i\}_{i=1}^n$ such that $\lambda_i\neq1$ for all $i$, the spectrum of $\widehat{\bbE}$ is given by
\begin{equation*}
\mathrm{Spec}(\widehat{\bbE}) = \left\{\prod_{1\leq j \leq n} \lambda_j^{n_j} \ep{\pm \iu \delta_{n_j,1} \theta_{0j}} : n_j=0,1,2\right\}\,.
\end{equation*}
Moreover, the largest eigenvalue $\prod_{i:\lambda_i>1} \lambda_i^2$ is non degenerate.
\end{lemma}
\begin{proof}
The proof is by induction on the number of particles. If only one particle is present, say with index $1$, the spectrum is easily seen to be exactly $\{1,\ep{\iu\theta_{01}}\lambda_1,\ep{-\iu\theta_{01}}\lambda_1, \lambda_1^2\}$. Let now $S$ be a non empty subset of $\{1,\ldots,n\}$, let $j\in S$, and denote $S^j = S\setminus \{j\}$. An arbitrary matrix $M\in\caM_{2^{\vert S\vert}}$ can be decomposed with respect to the tensor factor corresponding to $j$ as $M = M_+\otimes \sigma^+ + M_-\otimes \sigma^- + M_P\otimes P + M_Q\otimes Q$, where $P=\sigma^+\sigma^-$ and $Q = 1-P$. A direct calculation yields
\begin{align*}
\widehat{\bbE}^S(M) &= \ep{-\iu\theta_{0j}}\lambda_j \widehat{\bbE}^{S^j}(M_+) \otimes \sigma^+
+ \ep{\iu\theta_{0j}}\lambda_j\widehat{\bbE}^{S^j}(M_-) \otimes \sigma^- \\
&\quad+ \widehat{\bbE}^{S^j}(M_P)\otimes P
+ \left((v_0^{S^j})\str M_P (v_0^{S^j}) + \lambda_j^2 \widehat{\bbE}^{S^j}(M_Q)\right)\otimes Q
\end{align*}
where $\widehat{\bbE}^S = \sum_{i\in S} (v_i^S)\str \cdot v_i^S$ and $v_i^S$ are the matrices of~(\ref{v0}, \ref{vi}), truncated to the $\abs{S}$ factors of $S$. The spectrum of $\widehat{\bbE}^S$ therefore contains the spectrum of $\widehat{\bbE}^{S^j}$ multiplied by $\ep{\pm\iu\theta_{0j}}\lambda_j$. Moreover, the eigenvalue equation in the $P,Q$ block can be cast as
\begin{equation}\label{FRec}
\begin{pmatrix} 
\widehat{\bbE}^{S^j} - \mu & 0 \\ (v_0^{S^j})\str \cdot (v_0^{S^j}) & \lambda_j^2 \widehat{\bbE}^{S^j} - \mu
\end{pmatrix} = 0\,,
\end{equation}
the determinant of which is simply $\det(\widehat{\bbE}^{S^j} - \mu)\det(\lambda_j^2 \widehat{\bbE}^{S^j} - \mu)$. Hence, $\mu\in\mathrm{Spec}(\widehat{\bbE}^{S})$ if either $\mu\in\mathrm{Spec}(\widehat{\bbE}^{S^j})$, or $\mu/\lambda_j^2\in\mathrm{Spec}(\widehat{\bbE}^{S^j})$. In summary,
%showing that $\mu\in\mathrm{Spec}(\widehat{\bbE}^{S})$ if either $\mu\notin\mathrm{Spec}(\widehat{\bbE}^{S^j})$ and $\mu/\lambda_j^2\in\mathrm{Spec}(\widehat{\bbE}^{S^j})$, or $\mu/\lambda_j^2\notin\mathrm{Spec}(\widehat{\bbE}^{S^j})$ and $\mu\in\mathrm{Spec}(\widehat{\bbE}^{S^j})$. Since Assumption~\ref{Assump:LambdaNonDeg} ensures that both cases are valid, 
we have that
\begin{equation}\label{SpecRec}
\mathrm{Spec}(\widehat{\bbE}^S) = \{1,\ep{\iu\theta_{0j}}\lambda_j,\ep{-\iu\theta_{0j}}\lambda_j, \lambda_j^2\}\cdot \mathrm{Spec}(\widehat{\bbE}^{S^j})\,,
\end{equation}
which closes the induction. 

Clearly, the largest eigenvalue is $\prod_{i:\lambda_i>1} \lambda_i^2$. Its non-degeneracy follows from the recursion~(\ref{SpecRec}) and the fact that $\lambda_i^2\neq 1$ for all $i$.
\end{proof}
In particular, $\widehat{\bbE}^S$ is injective for all $S$. It follows from~(\ref{FRec}) that $M_P=0$ whenever $\lambda_j^2$ is a factor of the corresponding eigenvector. Hence, if there is an index $i$ with $\lambda_i>1$, then the right eigenvector for the largest eigenvalue is not strictly positive. Similarly for the corresponding left eigenvector whenever there is a $j$ for which $\lambda_j<1$, and it is therefore not faithful as a state on $\caM_k$. In both cases, the representation is not minimal.

Rescaling $\bbE$ by the largest eigenvalue of the transfer operator yields a completely positive map satisfying the conditions of Section~\ref{sub:MPS} and generates a unique translation invariant state on the chain algebra. We shall now describe the limiting state as well as the ground state spaces of the half-infinite chains. First, we consider the asymptotics of the norms of $\psi_{[1,N]}^S$ as $N\to\infty$. It is straightforward to observe that the same holds for $\psi_{[-N+1,0]}^S$ by replacing the conditions $\lambda_i\lessgtr1$ by $\lambda_i\gtrless1$.
\begin{lemma} \label{lma:norms}
Let $\psi_N^S$ be the vectors given by~(\ref{GSvectors}). There exists a finite constant $C^S>0$ such that
\begin{equation*}
\lim_{N\to\infty}\frac{\left\Vert \psi_N^S \right\Vert^2}{\prod_{i\in S:\lambda_i>1} \lambda_i^{2N}} \longrightarrow C^S\,,
\end{equation*}
where the denominator is equal to $1$ if the product is empty.
\end{lemma}
\begin{proof}
Recalling~(\ref{ScalarProdGamma}), we have that
\begin{equation}
\left\Vert \psi_N^S \right\Vert^2 = \sum_{p,q=1}^k \braket{\chi_p}{ \widehat{\bbE}^{(N)}\left((B^{S})\str |\chi_p\rangle\langle\chi_q|B^{S}\right) \chi_q} 
= \sum_{k=1}^{2^n} t_k^N\Tr\left(L_k(B^{S})\str R_k B^{S}\right)\,, \label{overlap}
\end{equation}
where $\{\chi_p\}_{p=1}^k$ is an arbitrary orthonormal basis of $\bbC^k$, and $R_k,L_k$ are the right and left eigenvectors of the transfer operator for the eigenvalue $t_k$. Since all $R_k$'s are simple tensor products and with $B^S$ given in~(\ref{Bs}), it is immediate to observe that $(B^{S})\str R_k B^{S} = 0$ whenever $R_k$ is not diagonal. Moreover, $B^{S} = P_j B^{S}P_j$ iff $j\notin S$. But $P_j R_k P_j = 0$ if $\lambda_j^2$ is a factor of $t_k$. Hence, the only terms of the sum that do not vanish are those corresponding to the eigenvalues of the form $t_k = \prod_{j}\lambda_j^2$ where the product runs over any subset of $S$. In the $N\to\infty$ limit, the norm is therefore dominated by the eigenvalue $\prod_{i\in S:\lambda_i>1}\lambda_i^2$. Note that if $\{i\in S:\lambda_i>1\}=\emptyset$, all summands vanish exponentially but the one corresponding to the eigenvalue $1$.
\end{proof}

Let
\begin{equation*}
\mean{A}{\psi} = \frac{\braket{\psi}{A\psi}}{\left\Vert \psi \right\Vert^2}
\end{equation*}
denote the expectation value of a local observable $A$ in the state $\psi$.
\begin{prop} \label{prop:ThermoLimits}
Let $A$ be a local observable, namely $A\in\caB(\caH_{[a,b]})$ for some $a<b\in\bbZ$. Then, for any $S\in\caP(\{1,\ldots,n\})$,
\begin{align*}
&\lim_{x,y\to\infty}\left\vert \mean{A}{\psi_{[a-x,b+y]}^S} - \mean{A}{\Omega}\right\vert = 0 \\
\text{for any }x\geq 0,\quad &\lim_{y\to\infty}\left\vert \mean{A}{\psi_{[a-x,b+y]}^S} - \mean{A}{\psi_{[a-x,b+y]}^{S_L}}\right\vert = 0\\
\text{for any }y\geq 0, \quad&\lim_{x\to\infty}\left\vert \mean{A}{\psi_{[a-x,b+y]}^S} - \mean{A}{\psi_{[a-x,b+y]}^{S_R}}\right\vert = 0 \\
\end{align*}
where $S_L = \{i\in S:\lambda_i<1\}$ and $S_R = \{i\in S:\lambda_i>1\}$, with the convention $\psi^{\emptyset}_{\cdots}\equiv\Omega$.
\end{prop}
\begin{proof}
For simplicity, we consider the case of $S =\{1\}$, and we drop the index $1$ everywhere. The general case follows similarly using Lemma~\ref{lma:norms}. We have
\begin{equation*}
\mean{A}{\psi_{[a-x,b+y]}} = \left\Vert \psi_{[a-x,b+y]} \right\Vert^{-2}\sum_{i,j=a-x}^{b+y} \ep{-\iu\theta(i-j)}\lambda^{i+j}\braket{e_0\cdots  e_1\cdots e_0}{A e_0\cdots  e_1\cdots e_0}
\end{equation*}
where $e_1$ is at position $i$ on the left and $j$ on the right. All summands where one index belongs to $[a,b]$ and the other does not vanish. Therefore,
\begin{align*}
\mean{A}{\psi_{[a-x,b+y]}} &= \left\Vert \psi_{[a-x,b+y]} \right\Vert^{-2} \Bigg[ \sum_{i\notin[a,b]} \lambda^{2i} \braket{e_0\cdots e_0}{A_{[a,b]} e_0\cdots e_0} \\
&\quad + \sum_{i,j\in[a,b]} \ep{-\iu\theta(i-j)}\lambda^{i+j} \braket{e_0\cdots e_1\cdots e_0}{A_{[a,b]} e_0\cdots e_1\cdots e_0} \Bigg] = T_1 + T_2
\end{align*}
where $A_{[a,b]}$ is the restriction of $A$ to the interval $[a,b]$. Now recall that $\left\Vert \psi_{[a-x,b+y]} \right\Vert^2 = \sum_{i=a-x}^{b+y} \lambda^{2i}$, which diverges as $x\to\infty$ if $\lambda<1$ and as $y\to\infty$ if $\lambda>1$. Since the interval $[a,b]$ is fixed and finite, $\lim_{x,y\to\infty}\abs{T_2} = 0$, independently of the value of $\lambda$, and further
\begin{multline*}
\lim_{x,y\to\infty}\left\vert \mean{A}{\psi_{[a-x,b+y]}} - \mean{A}{\Omega}\right\vert 
\\
\leq \abs{\braket{e_0\cdots e_0}{A_{[a,b]} e_0\cdots e_0}} \lim_{x,y\to\infty}  \abs{1-\left\Vert \psi_{[a-x,b+y]} \right\Vert^{-2} \sum_{i\notin[a,b]} \lambda^{2i}} + \lim_{x,y\to\infty}\abs{T_2} = 0\,,
\end{multline*}
which proves the first statement of the proposition. We now assume that $\lambda<1$, the other case is the same after exchanging the roles of the left and the right boundaries. As $x\to\infty$, the calculation above continues to hold with $y$ fixed. Finally, if $x$ is fixed, $\lim_{y\to\infty}\left\Vert \psi_{[a-x,b+y]} \right\Vert^2 = (1-\lambda^{2i})^{-1}$ so that both $\lim_{y\to\infty}\abs{T_1}$ and $\lim_{y\to\infty}\abs{T_2}$ are finite, non-vanishing constants.
\end{proof}
Summarizing, there is a unique ground state in the thermodynamic limit, \ie on $\bbZ$, which is a simple product state $\Omega$. There is no correlation between spins at different sites of the chain. For each of the two half-infinite chains, the ground states are determined by their particle content, with only left particles appearing on the right infinite chain with a left boundary and symmetrically for the left infinite chain. The binding of the particles to the edges is reflected in the non trivial spectrum of the transfer operator. There are $2^{n_L}$, resp. $2^{n_R}$, ground states on $[1,+\infty)$, resp. $(-\infty,0]$. As can be read from the proof, the convergence to the limiting ground states is exponential, with rate at least
\begin{equation}\label{RateOfConvergence}
\varepsilon = \mathrm{max}\left\{\mathrm{min}\left\{\lambda_i, \lambda_i^{-1}\right\}: 1\leq i \leq n\right\}\,.
\end{equation}

Note that it would be wrong to conclude that these particles are massless. In fact, there is a mass gap for each of them when they are in the bulk, see~(\ref{MassGap}) and its proof below, but the binding energy to the edge compensates exactly for it as long as no other particle of that type is already bound.

The rest of this section is devoted to the proofs of Proposition~\ref{Prop:PVBS_Gap} and Theorem~\ref{thm:PVBS_phases}. We start with the lemma below which is a simple but crucial property relating a ground state of the long chain to the ground states of the two subsystems obtained by truncating the chain.

\begin{lemma} \label{lma:truncate}
For any $S\subseteq\{1,\ldots,n\}$ and $0 < k < N$, 
\begin{equation*}
\psi_N(S) = \sum_{\substack{S_k\subseteq S \\ \abs{S_k}\leq k}} T_{N,k}(S_k) \psi_{N-k}(S\setminus S_k) \otimes \psi_{k}(S_k)\,,
\end{equation*}
where
\begin{equation*}
\abs{T_{N,k}(S_k)} = \prod_{i\in S_k}\lambda_i^{N-k}\,.
\end{equation*}
\end{lemma}
\begin{proof}
The decomposition follows from the intersection property and the fact that any particle of $S$ can appear at most once in $\psi_N(S)$. The position of any particle of $S_k$ in $\psi_N(S)$ is $N-k$ steps further from the left boundary than in $\psi_{k}(S_k)$, so that each of them is associated with an additional factor $\lambda_i^{N-k}$, up to phases.
\end{proof}

\begin{proof}[Proof of Proposition~\ref{Prop:PVBS_Gap}] \emph{Part i: Lower bound.} 
First we note that the intersection property implies that
\begin{equation*}
G_{[1,N]}-G_{[1,N+1]} = (\idtyty - G_{[1,N+1]})G_{[1,N]}\,,
\end{equation*}
so that a vector $\Phi$ the range of $G_{[1,N]}-G_{[1,N+1]}$ is characterized equivalently by
\begin{equation*}
\Phi\in \caG_{[1,N]} \cap \caG_{[1,N+1]}^\perp\,.
\end{equation*}
Then,
\begin{equation*}
g_{k,N}^2 = \sup_{\substack{\Phi\in \caG_{[1,N]} \cap \caG_{[1,N+1]}^\perp \\ \Vert\Phi\Vert = 1}}\left\Vert G_{[N-k+2, N+1]}\Phi\right\Vert^2 = \sup_{\substack{\Phi\in \caG_{[1,N]} \cap \caG_{[1,N+1]}^\perp \\ \Vert\Phi\Vert = 1}}\braket{G_{[N-k+2, N+1]}\Phi}{\Phi}\,.
\end{equation*}
$\Phi\in\caG_{[1,N]}$ implies $\Phi\in\caG_{[1,N-k+1]}$ by the intersection property. Moreover, $H_{[1,N-k+1]}$ commutes with $G_{[N-k+2, N+1]}$ so that
\begin{equation*}
G_{[N-k+2, N+1]}\Phi \in \caG_{[1,N-k+1]} \cap \caG_{[N-k+2, N+1]} = \caG_{[1,N-k+1]} \otimes \caG_{[N-k+2, N+1]}\,.
\end{equation*}
Hence,
\begin{multline} \label{PVBS:start}
g_{k,N}^2 \leq \sup_{\Phi\in \caG_{[1,N+1]}^\perp, \Vert\Phi\Vert = 1} \sum_{P, P', S\subset\{1,\ldots,n\}}\sum_{i=0}^n \\
\Bigg\vert \frac{\overline{\phi_{k,N}(P,P')}\varphi_N(S,i)}{\Vert \psi_{N-k+1}(P) \Vert \Vert \psi_{k}(P') \Vert\Vert \psi_N(S) \Vert }\braket{\psi_{N-k+1}(P)\otimes\psi_{k}(P')}{\psi_N(S)\otimes e_i} \Bigg\vert\,,
\end{multline}
where $\phi_{k,N}(P,P')= \langle\hat{\psi}_{N-k+1}(P)\otimes\hat{\psi}_{k}(P'),\Phi\rangle$ and $\varphi_N(S,i)= \langle\hat{\psi}_N(S)\otimes e_i,\Phi\rangle$. The normalization $\Vert \Phi \Vert = 1$ and the orthogonality of the vectors $\{\psi_N(S)\}_S$ implies that $\vert \varphi_N(S,i) \vert \leq 1$ and $\vert \phi_{k,N}(P,P') \vert \leq 1$. Moreover, $\Phi \in \caG_{[1,N+1]}^\perp$, so that
\begin{equation}\label{PVBS:orthogonality}
0=\sum_{S',i} \frac{\varphi_N(S',i)}{\Vert \psi_N(S') \Vert \Vert \psi_{N+1}(S) \Vert} \braket{\psi_{N+1}(S)}{\psi_N(S') \otimes e_i} 
\end{equation}
for all $S$. By Lemma~\ref{lma:truncate}, this reads
\begin{align*}
0
&= \sum_{S',i} \sum_{j\in S\cup\{0\}} \frac{\varphi_N(S',i)}{\Vert \psi_N(S') \Vert \Vert \psi_{N+1}(S) \Vert} T_{N+1,1}(j)\braket{\psi_{N}(S^j)\otimes e_j}{\psi_N(S') \otimes e_i} \\
&= \sum_{j\in S\cup\{0\}} \varphi_N(S^j,j) \frac{ \Vert \psi_N(S^j) \Vert}{\Vert \psi_{N+1}(S) \Vert} T_{N+1,1}(j)\,.
\end{align*}
Denoting by $S_{L/R} :=S\cap\Lambda_{L/R}$ and observing that $\Vert \psi_N(S^j) \Vert / \Vert \psi_{N+1}(S) \Vert \, T_{N+1,1}(j) = \caO(\lambda_i^N)$ if $i\in S_L$ by Lemmas~\ref{lma:norms} and~\ref{lma:truncate}, we obtain
\begin{equation} \label{PVBS:NimitOrth}
\abs{\sum_{j\in S_R\cup\{0\}} \varphi_N(S^j;j) \frac{ \Vert \psi_N(S^j) \Vert}{\Vert \psi_{N+1}(S) \Vert} T_{N+1,1}(j)} 
\leq C\varepsilon_L^{N}
\end{equation}
for any $S$ (it vanishes if $S_L=\emptyset$), where $C$ is a constant and
\begin{equation*}
\varepsilon_L = \mathrm{max}\left\{\lambda_i:i \in \Lambda_L \right\}\,.
\end{equation*}
Furthermore,
\begin{equation*}
\braket{\psi_{N-k+1}(P)\otimes\psi_{k}(P')}{\psi_{N}(S)\otimes e_i}
= 1_{i\in P'} 1_{S=P\cup{P'}^i} T_{k,1}(i) T_{N,k-1}({P'}^i) \Vert \psi_{N-k+1}(P) \Vert^2 \Vert \psi_{k-1}({P'}^i) \Vert^2\,.
\end{equation*}
Hence,
\begin{equation*}
g_{k,N}^2 \leq \sup_\Phi \sum_{P,P'}\sum_{i\in P'}  \abs{\overline{\phi_{k,N}(P,P')}\varphi_N(P\cup {P'}^i,i) \frac{\Vert \psi_{N-k+1}(P) \Vert \Vert \psi_{k-1}({P'}^i)\Vert^2}{\Vert \psi_{k}(P') \Vert \Vert \psi_{N}(P\cup{P'}^i) \Vert} T_{k,1}(i) T_{N,k-1}({P'}^i)}
\end{equation*}
Asymptotically,  for $N\geq 2k-1$, the coefficients are
\begin{equation*}
\caO\left(\frac{\prod_{j\in P'\cap\Lambda_L}\lambda_j^k}{\prod_{j\in P\cap\Lambda_R}\lambda_j^k}\right)
\end{equation*}
and hence vanish as $C_1\varepsilon^{k}$ with $\varepsilon$ defined in~(\ref{RateOfConvergence}) in all cases but $P\cap\Lambda_R = \emptyset = P'\cap\Lambda_L$, and in particular $i\in \Lambda_R\cup\{0\}$. The latter terms can be recast as follows
\begin{multline*}
g_{k,N}^2 \leq C_1\varepsilon^k \\
+C_2\sup_\Phi \sum_{S} \abs{\sum_{i\in S_R\cup\{0\}} \varphi_N(S^i,i) \frac{\Vert \psi_{k-1}(S_R^i)\Vert}{\Vert \psi_{k}(S_R) \Vert} T_{k,1}(i)}
\frac{\Vert \psi_{N-k+1}(S\setminus S_R) \Vert \Vert \psi_{k-1}(S_R^i)\Vert \abs{T_{N,k-1}(S_R^i)}}{ \Vert \psi_{N}(S^i) \Vert}\,.
\end{multline*}
It remains to observe that the second quotient is uniformly bounded and that the sum in the absolute value is exponentially small for any $\Phi$ by the orthogonality condition~(\ref{PVBS:NimitOrth}), as
\begin{equation*}
\frac{\Vert \psi_{N}(S^i) \Vert}{\Vert \psi_{N+1}(S) \Vert} T_{N+1,1}(i) = \frac{\Vert \psi_{k-1}(S_R^i)\Vert}{\Vert \psi_{k}(S_R)\Vert}T_{k,1}(i)(1+\caO(\varepsilon^{k-1}))\end{equation*}
for any $N\geq 2k-1$. Finally, since $\varepsilon_L\leq \varepsilon$, there exists a finite constant $K$ such that
\begin{equation*}
g_{k,N}^2 \leq K \varepsilon^{k-1}<1/k\,.
\end{equation*}
for $k$ sufficiently large and all $N\geq 2k-1$. This yields a uniform lower bound on the spectral gap by the remarks in Section~\ref{sub:MPS}.
%%%%%%%%%%%%%%%

\emph{Part ii: Upper bound.} We shall use the variational principle
\begin{equation*}
\gamma(\Lambda, \Theta, N) = \inf_{0\neq \psi\perp\mathrm{Ker}(H_{[a,b]})}\left\langle H_{[a,b]}\right\rangle_\psi = \inf_{\psi\notin\mathrm{Ker}(H_{[a,b]})}\left\langle H_{[a,b]}\right\rangle_{\psi-G_{[a,b]}\psi} = \inf_{\psi\notin\mathrm{Ker}(H_{[a,b]})}\frac{\braket{\psi}{H_{[a,b]}\psi}}{\Vert(1-G_{[a,b]})\psi\Vert^2}\,.
\end{equation*}
First, we observe that subspaces with a fixed number of particles are orthogonal to each other and invariant for the Hamiltonian. Let us consider $H_{[a,b]}^i$, the Hamiltonians restricted to the subspace containing exactly one particle of type $i$ and none other. In order to study its spectrum, we center and rescale it to
\begin{equation*}
K_{[a,b]}^i := -(\lambda_i+\lambda_i^{-1})(H_{[a,b]}^i-1)
\end{equation*}
which has matrix elements 
\begin{equation*}
(K_{[a,b]}^i)(x,y) = \begin{cases} 
\ep{-\iu\theta_{i0}}\delta_{x+1,y} + \ep{\iu\theta_{i0}}\delta_{x-1,y} & a < x < b\,, \\
\lambda_i^{-1} \delta_{a,y} + \ep{-\iu\theta_{i0}}\delta_{a+1,y} & x=a\,, \\
\ep{\iu\theta_{i0}}\delta_{b-1,y} + \lambda_i \delta_{b,y} & x=b\,.
\end{cases}
\end{equation*}
On a finite chain $[a,b]$ of length $N$, let $w$ be a variational vector of the form $w(x) = \exp(\iu k x)$ for $a\leq x\leq b$ and any $k\in\bbR$. Then,
\begin{equation*}
\braket{w}{K_{[a,b]}^i w} = (\lambda_i^{-1} + \lambda_i)+(N-1)E(k)\,,\end{equation*}
where $E(k) = 2\cos(\theta_{i0}-k)$. Moreover,
\begin{equation*}
\Vert(1-G_{[a,b]}^i)w\Vert^2 = \Vert w \Vert^2 - \Vert G_{[a,b]}^i w \Vert^2 = N-\frac{1}{\Vert \psi_{[a,b]}^i \Vert^2}\abs{\braket{\psi_{[a,b]}^i}{w}}^2=: N-C_N(k)
\end{equation*}
where $\psi_{[a,b]}^i$ is the ground state given in~(\ref{iGS}). In terms of the actual Hamiltonian, we have
\begin{equation}\label{ExpHam_i}
\frac{\braket{w}{H_{[a,b]}^iw}}{\Vert(1-G_{[a,b]}^i)w\Vert^2} = \frac{N-1}{N-C_N(k)} \left(1-\frac{E(k)}{\lambda_i^{-1} + \lambda_i}\right)
\,.
\end{equation}
Since both $\vert\langle\psi_{[a,b]}^i,w\rangle\vert$ and $\Vert \psi_{[a,b]}^i \Vert^2$ are geometric series, of ratio $\lambda_i\exp{\iu(k-\theta_{i0})}$ respectively $\lambda_i^2$, the quotient $C_N(k)$ is convergent as $N\to\infty$ for any $\lambda_i\neq1$. Since~(\ref{ExpHam_i}) holds for any $k$, the least upper bound is obtained for $k-\theta_{i0} = 0$, in which case $E(\theta_{i0}) = 2$ and
\begin{equation*}
C_N(\theta_{i0}) = \frac{(1+\lambda_i)}{(1-\lambda_i)}\frac{(1-\lambda_i^N)}{(1+\lambda_i^N)}
\leq \frac{(1+\lambda_i)}{\abs{1-\lambda_i}}=C_i
\end{equation*}
for all $N$. Hence,
\begin{equation*}
\frac{\braket{w}{H_{[a,b]}^iw}}{\Vert(1-G_{[a,b]}^i)w\Vert^2} \leq \left(1-\frac{2}{\lambda_i^{-1} + \lambda_i}\right)\left(1 + \frac{C_i-1}{N-C_i}\right)\,.
\end{equation*}
\end{proof}

With the normalization $\lambda_0=1$, the assumption that $\lambda_i\neq1$ for all $i\neq0$ is essential in the proof of the existence of a spectral gap. In fact, the variational upper bound shows that the gap closes indeed as $\vert\lambda_i-1\vert\to 0$ for some $i\neq 0$.

The upper bound of Proposition~(\ref{Prop:PVBS_Gap}) is sharp in each invariant subspace, in the sense that $1-2/(\lambda_i^{-1} + \lambda_i)$ also a strict lower bound for finite $N$, and therefore gives the exact value of the gap in the thermodynamic limit for the Hamiltonian restricted to the one-particle spaces. We continue working with $K^i$ for which the spectral picture there is inverted: The minimal eigenvalue $0$ for the Hamiltonian $H^i$ is the maximal one $\lambda_i+\lambda_i^{-1}$ for its transformed version $K^i$, and the gap bound $1-2/(\lambda_i+\lambda_i^{-1})$ for $H^i$ corresponds to $2$ for $K^i$. Let us write the eigenvalue equation $K_{N+2}^i \psi = E \psi$ using the transfer matrix
\begin{equation*}
T(E) = \begin{pmatrix}
\ep{\iu\theta_{i0}} E & -\ep{2\iu\theta_{i0}} \\ 1 & 0 \end{pmatrix} = \begin{pmatrix}
\ep{\iu\theta_{i0}} & 0 \\ 0 & 1 \end{pmatrix}
\begin{pmatrix}
E & -1 \\ 1 & 0 \end{pmatrix}
\begin{pmatrix}
1 & 0 \\ 0 & \ep{\iu\theta_{i0}} \end{pmatrix}\,.
\end{equation*}
The boundary conditions
\begin{equation*}
\psi(1) = \ep{\iu\theta_{i0}}(E-\lambda_i^{-1})\psi(0) \,,\qquad \psi(N) = \ep{-\iu\theta_{i0}}(E-\lambda_i)\psi(N+1)\,,
\end{equation*}
yield the equation
\begin{equation} \label{TransferEq}
\begin{pmatrix}
1 \\ E-\lambda_i \end{pmatrix} \psi_{out} = \ep{\iu (N+1) \theta_{i0}} \begin{pmatrix}
E & -1 \\ 1 & 0 \end{pmatrix}^{N}
\begin{pmatrix}
E-\lambda_i^{-1} \\ 1 \end{pmatrix} \psi_{in}
\end{equation}
on the chain of length $N+2$. Since
\begin{equation*}
\begin{pmatrix}
E & -1 \\ 1 & 0 \end{pmatrix}^N = 
\begin{pmatrix}
U_N(E/2) & -U_{N-1}(E/2) \\ U_{N-1}(E/2) & -U_{N-2}(E/2)
\end{pmatrix}
\end{equation*}
where $U_N(x)$ are the Chebyshev polynomials of the second kind, (\ref{TransferEq}) reduces to
\begin{equation*}
\left( E-\lambda_i \right)\left[U_{N}(E/2)\left( E-\lambda_i^{-1} \right) - U_{N-1}(E/2) \right]  
= U_{N-1}(E/2)\left( E-\lambda_i^{-1} \right) - U_{N-2}(E/2)\,,
\end{equation*}
and further to
\begin{equation}
\label{Cheb}
U_{N}(E/2)
= \left[\frac{1}{ \left( E-\lambda_i \right)} + \frac{1}{\left( E-\lambda_i^{-1} \right) } \right] U_{N-1}(E/2) - \frac{1}{\left( E-\lambda_i^{-1} \right) \left( E-\lambda_i \right)}U_{N-2}(E/2)\,.
\end{equation}
If $E = \lambda_i+\lambda_i^{-1}\geq 2$, this equation is precisely the defining recursion relation for the Chebyshev polynomials, so that $\lambda_i+\lambda_i^{-1}$ is an eigenvalue for all $N$, and indeed the largest one. Moreover, as $U_N(1) = N+1$, we also note that $E = 2$ does not belong to the spectrum for any finite $N$, as~(\ref{Cheb}) reads
\begin{equation*}
1 = \frac{3-(\lambda_i+\lambda_i^{-1})}{5-2(\lambda_i+\lambda_i^{-1})}
\end{equation*}
whose only real solution is $\lambda_i=1$, which is excluded. Now, we can use the eigenvalues of the transfer matrix
\begin{equation*}
t_\pm = \frac{1}{2}\left(E\pm\sqrt{E^2-4}\right)
\end{equation*}
to diagonalize it away from $|E|=2$ and obtain the following form of~(\ref{TransferEq}):
\begin{align*}
\left(1-t_-\left(E-\lambda_i\right)\right)\psi_{out} &= \ep{i(N+1)\theta_{i0}}t_+^N\left((E-\lambda_i^{-1})-t_-)\right)\psi_{in} \\
\left(-1+t_+\left(E-\lambda_i\right)\right)\psi_{out} &= \ep{i(N+1)\theta_{i0}}t_-^N\left(-(E-\lambda_i^{-1})+t_+)\right)\psi_{in}
\end{align*}
If $E \neq \lambda_i+\lambda_i^{-1}$, the system yields
\begin{equation*}
\frac{1-t_- (E-\lambda_i)}{1-t_+(E-\lambda_i)}\frac{t_+ - (E-\lambda_i^{-1})}{t_- -(E-\lambda_i^{-1})} 
 = \left(\frac{t_+}{t_-}\right)^{N}\,.
\end{equation*}
The left hand side is strictly decreasing for $E\in[2, \lambda_i+\lambda_i^{-1})$ and its limiting value at $E=2$ is $1$, so that
\begin{equation*}
1 > \frac{1-t_- (E-\lambda_i)}{1-t_+(E-\lambda_i)}\frac{t_+ - (E-\lambda_i^{-1})}{t_- -(E-\lambda_i^{-1})}
 = \left(\frac{t_+}{t_-}\right)^{N}\,,
\end{equation*}
which is a contradiction as $t_+>t_-$ for $E>2$. Hence, there are no eigenvalues of $K_{N+2}^i$ in the interval $[2, \lambda_i+\lambda_i^{-1})$. Therefore, if $\gamma^i(N)$ is the spectral gap above the ground state of $H^i_N$, we have that
\begin{equation*}
\inf_{N\in\bbN}(\gamma^i(N)) \geq 1-\frac{2}{\lambda_i+\lambda_i^{-1}}\,.
\end{equation*}
We further conjecture that
\begin{equation*}
\min\left\{ 1-\frac{2}{\lambda_i+\lambda_i^{-1}} : 1\leq i \leq n\right\}
\end{equation*}
is in fact the exact gap of the full system in the thermodynamic limit.

\begin{proof}[Proof of Theorem~\ref{thm:PVBS_phases}]
We recall that if two models are in the same gapped ground state phase, then for any lattice under consideration their ground state spaces are mapped onto each other by a quasi-local automorphism. Suppose that $n_L^1 \neq n_L^2$. Then the dimension of the ground state spaces of the two PVBS models on the half-infinite chain with a left boundary are not equal, and the two sets cannot be related by an automorphism. The same holds for $n_R^1\neq n_R^2$ and an infinite chain with right boundary. Hence, the condition $(n^1_L, n^1_R) = (n^2_L, n^2_R)$ is necessary for equivalence. In order to show that it is also sufficient, we construct a smooth interpolating path of gapped Hamiltonians and conclude by~\cite{BMNS}. The parameters of the two PVBS models are denoted by $(\{\lambda_i\},\{\theta_{ij}\})$, respectively $(\{\mu_i\},\{\varphi_{ij}\})$. By relabeling the basis elements $e_i$, we can assume that
\begin{equation*}
\lambda_1\leq \lambda_2\leq\cdots\leq\lambda_{n_L}<1<\lambda_{n_L+1} \leq \cdots\leq\lambda_{n_L+n_R}\,.
\end{equation*}
Let $(i_1,\ldots,i_{n})$ be the sequence of indices for such that
\begin{equation*}
\mu_{i_1}\leq \cdots\leq\mu_{i_{n_L}}<1<\mu_{i_{n_L+1}}\leq \cdots\leq\mu_{i_{n_L+n_R}}\,.
\end{equation*}
Let $u$ be the unitary map defined by
\begin{equation*}
ue_{i_j} = e_j\,,\qquad ue_0 = e_0\,,
\end{equation*}
and $u(s)\in SU(n+1)$ a smooth path for $s\in[0,1]$, such that $u(1)=1$ and $u(0) = u$. Finally, let $U(s) = \otimes_{x=a}^b u(s)$. In the same interval, let
\begin{align*}
\lambda_j(s) &= s\lambda_j + (1-s)\mu_{i_j} \\
\theta_{jk}(s) &= s\theta_{jk} + (1-s)\varphi_{i_ji_k}
\end{align*}
and $h(s)$ be the PVBS interaction constructed from these $s$-dependent parameters. Then the Hamiltonian
\begin{equation*}
H_{[a,b]}(s):= U(s) \left[\sum_{x=a}^{b-1} h(s)_{x,x+1} \right] U(s)\str
\end{equation*}
satisfies all the properties required for the construction of a quasi-local automorphism. Clearly, $H_{[a,b]}(0) = H_{[a,b]}^{\mu,\varphi}$ and $H_{[a,b]}(1) = H_{[a,b]}^{\lambda,\theta}$. Moreover, the linear interpolation of the parameters generates a smooth family of normalized ground state vectors $\hat\psi_{[a,b]}^S(s)$, which remain orthogonal to each other for $s\in[0,1]$. This implies that the path $h(s)$ is smooth as well. Since all paths $\lambda_j(s)$ interpolate linearly between parameters that are either both larger or both smaller than $1$, we have that $\abs{\lambda_j(s)-1}>0$ for all $j$ and $s\in[0,1]$. Hence, Proposition~\ref{Prop:PVBS_Gap} ensures that the Hamiltonian $\sum_{x=a}^{b-1} h(s)_{x,x+1}$ is gapped, uniformly in the length $\abs{b-a+1}$. Finally, the unitary conjugation preserves the spectrum and therefore the gap, and since $U(s)$ is local, the Hamiltonian $H_{[a,b]}(s)$ has a nearest neighbor interaction.
\end{proof}
%

%%%%%%%%%%%%%%%%%%%%%%%%%%%%%%%%%%%%%%%%%%%%%%%%%%%%%%%%%%%%%%%%%%%%%

\section{The ground state phase of the AKLT model}

%%%%%%%%%%%%%%%%%%%%%%%%%%%%%%%%%%%%%%%%%%%%%%%%%%%%%%%%%%%%%%%%%%%%%

In the previous section, we introduced a family of simple models that can be classified according to their boundary behavior. They are all gapped with a unique ground state in the bulk, but they support a number of additional edge modes in their ground state spaces as soon as boundaries are present. In the bulk, these states become gapped excitations. 

We believe that these characteristics are the key for the classification of gapped phases beyond symmetry breaking in one dimension. We now illustrate the use of the PVBS classes with a physically relevant example: The AKLT model~\cite{AKLT}. It is a one-dimensional spin-$1$ chain, namely $\caH_x = \bbC^3$ carrying the three dimensional representation of $\mathfrak{su}(2)$, with generators denoted $(S^1, S^2, S^3)$. Its Hamiltonian has a nearest-neighbor interaction
\begin{equation*}
H_{AKLT} = \sum_x \left[ \frac{1}{2} \left(S_x\cdot S_{x+1}\right) + \frac{1}{6}\left(S_x\cdot S_{x+1}\right)^2 + \frac{1}{3} \right]
\end{equation*}
corresponding to the projection onto the spin-$2$ subspace of every pair of neighbouring spins $\bbC^3\otimes\bbC^3$. As such, $H_{AKLT}\geq 0$. Its unique ground state in the thermodynamic limit has a matrix product representation identified in~\cite{FCS} described by the set of matrices
\begin{equation} \label{AKLTws}
w_{-} = \begin{pmatrix}
0 & -\sqrt{2/3} \\  0 & 0
\end{pmatrix}\,,\qquad
w_0 = \begin{pmatrix}
-1/\sqrt{3} & 0 \\  0 & 1/\sqrt{3}
\end{pmatrix}\,,\qquad
w_+ = \begin{pmatrix}
0 & 0 \\ \sqrt{2/3} & 0
\end{pmatrix}\,,
\end{equation}
in a fixed basis $(\chi_l)_{l=\pm}$ of $\bbC^2$. They satisfy the following commutation relations
\begin{gather}\label{AKLTCommutation}
w_0w_- = -w_-w_0\,,\qquad w_0w_+ = -w_+w_0\,,\qquad w_\pm^2 = 0\,,Ê\\
w_-w_+ + w_+w_- = -2w_0^2\,.
\end{gather}
For notational clarity, we denote here the transfer operator and the matrix product map on $\caM_2$ by $\widehat{\bbT}$ and $\Upsilon_N$. The uniqueness of the thermodynamic ground state follows from the fact that $1$ is an isolated eigenvalue. Moreover, and unlike the PVBS situation, there are short range correlations between sites in the thermodynamic limit corresponding to the triply degenerate eigenvalue $-1/3$.

The four-dimensional ground state space of a finite chain, namely the range of $\Upsilon_N$, has the following picture: It is isomorphic to the tensor product of two Bloch spheres that can be thought of as sitting at the two boundaries of the chain. Despite the superficial simplicity of this description, the ground states have a complicated structure that can be understood as follows. In the product basis, the only words with non vanishing coefficients in any ground state are those with alternating `$+$' and `$-$', separated by an arbitrary number of `$0$'. The four ground states differ by whether the first non zero letter is a `$+$' or a `$-$' and whether they come in equal numbers or with one additional `$+$' or `$-$'. In particular, they do not differ in the bulk and they all converge to the same weak-* limit point in the thermodynamic limit. Moreover, the system on the half-infinite line has a two-dimensional ground state space isomorphic to a single Bloch sphere. For a more detailed description of the model, but in the slightly different language of valence bond states, we refer to the original article~\cite{AKLT}. Furthermore, \cite{String, Hidden} identify a `hidden string order' as a rigorous description of the `dilute Neel order' heuristically introduced above.

With this quick introduction, we can now state the main result of this section.
\begin{thm} \label{thm:AKLT_phase}
The AKLT model belongs to the PVBS phase with $n_L = n_R = 1$.
\end{thm}
In other words, there exists a quasi-local automorphism of the observable algebra that maps the ground state spaces of the AKLT model on finite, half-infinite and infinite chains onto those of the PVBS class with one left and one right boundary particle. In particular, the intricate bulk state of the AKLT chain is equivalent to the simple product state common to all PVBS models. 

The importance of the locality property of the automorphism comes into sharp focus when one considers the ground states on the infinite and half-infinite chains. In finite volume, any unitary will preserve the dimension of the ground state space. The non-local unitary transformation introduced by Kennedy and Tasaki in~\cite{Hidden} to reveal the hidden string order transforms the 4 AKLT  ground states into 4 translation invariant product states, which leads to 4 distinct bulk ground states in the thermodynamic limit. This example shows that non-local unitary transformations do not preserve the structure of the bulk ground state(s) and clearly would not be useful to classify gapped ground state phases.

To prove the quasi-local automorphic equivalence of Theorem~\ref{thm:AKLT_phase}, it suffices again to construct a smooth path of uniformly gapped Hamiltonians interpolating between the two models. We shall first construct a path of interactions, show that the Hamiltonians have matrix product ground states and use this property to prove the existence of a gap.

In~\cite{PhysRevB.86.035149}, we presented the result without complete proof. In order to make the comparison with the current article easier, we note the change of notation from $\lambda(s)$ and $\mu(s)$ there to $f(s)$ and $g(s)$ here.

%%%%%%%%%%%%

\subsection{A path of Hamiltonians}
\label{sub:HamPath}

In order to obtain a smooth path of Hamiltonians, we define a smooth path of nearest neighbor, translation invariant interactions interpolating between the AKLT interaction and the PVBS interaction for $n_L = n_R = 1$. Both extreme points being projections, it suffices to define a smooth path of vectors, and to use their span as the range of a path of projections. Let $s_0$ be defined by $\sin(s_0) =\sqrt{2/3}$. For $s\in[0,s_0]$, let $h(s)$ be the projection onto the five dimensional space spanned by the vectors
\begin{equation} \label{G2}
\begin{split}
&e_+\otimes e_+ \,,\qquad e_-\otimes e_- \,,\qquad e_+\otimes e_0 + f(s) e_0\otimes e_+ \,,\qquad e_0\otimes e_- + f(s) e_-\otimes e_0 \,, \\
&e_+\otimes e_- + g(s)\frac{\sin(s)}{\cos^2(s)}e_0\otimes e_0 + f^2(s)e_-\otimes e_+\,,
\end{split}
\end{equation}
where
\begin{equation} \label{IsoCond}
|g(s)|^2 + |f(s)|^2\cos^2(s) = 1\,.
\end{equation}
\begin{lemma} \label{lem:HScont}
Assume that $s\mapsto f(s)$ is $C^1([0,s_0])$, that $|f|\leq 1$ and let $g(s)$ be a solution of~(\ref{IsoCond}). Then, for any $a<b\in\bbZ$, the map
\begin{equation*}
s\longmapsto H_{[a,b]}(s) = \sum_{x=a}^{b-1}h(s)_{x,x+1}
\end{equation*}
is of class $C^1([0,s_0]; \caB(\bbC^{3(b-a+1)}))$.
\end{lemma}
\begin{proof}
If $f(s)$ is $C^1([0,s_0])$ then all of the vectors above are strongly continuously differentiable in $[0,s_0]$, and since they are orthogonal to each other for all $s$, the projector $h(s)$ is continuously differentiable in $[0,s_0]$ in the strong operator topology. Hence $s\mapsto H_{[a,b]}(s)$ is strongly $C^1$, and therefore also in norm for all $a<b\in\bbZ$.
\end{proof}
Let now $f$ be an monotonely increasing real function satisfying the boundary conditions
\begin{equation*}
f(0) = 1/\sqrt{2}\,,\qquad f(s_0) = 1\,,
\end{equation*}
At $s=s_0$, the interaction is the AKLT interaction. At $s=0$ on the other hand, the range~(\ref{G2}) of the interaction corresponds precisely to the vectors~(\ref{phiij},\ref{phiii}) orthogonal to the ground state space of the PVBS model with parameters
\begin{equation*}
\lambda_+ = \sqrt{2}\,,\qquad \lambda_-=1/\sqrt{2}\,,
\end{equation*}
and
\begin{equation*}
\theta_{+-} = \theta_{0+} = \theta_{0-} = \pi\,.
\end{equation*}
Hence, the Hamiltonian at $s=0$ is a PVBS model with $n_L = n_R = 1$ and the path of Hamiltonians defined in Lemma~\ref{lem:HScont} interpolates as announced. It remains to show that these operators have a uniform spectral gap above the ground states.

%%%%%%%%%%%%

\subsection{A path of matrix product states}
\label{sub:PathFar}

In this section, we show that the Hamiltonians can be constructed pointwise for $s>0$ as the parent Hamiltonians (see~\cite{VBSHam}) of an interpolating family of matrix product states. For $s\in[0,s_0]$, we define the $2\times 2$ matrices
\begin{equation} \label{wMatrices}
w_{-}(s) = \begin{pmatrix}
0 & -g(s) \\  0 & 0
\end{pmatrix}\,,\qquad
w_0(s) = -\cos(s)\begin{pmatrix}
1 & 0 \\ 0 & -f(s)
\end{pmatrix}\,,\qquad
w_+(s) = \begin{pmatrix}
0 & 0 \\ \sin(s) & 0
\end{pmatrix}\,,
\end{equation}
where $f(s)$ and $g(s)$ were introduced above. For all $s$, let $\bbT_A(s;B)$ be the completely positive map corresponding to these matrices.
\begin{lemma}\label{lem:MPS}
Let $f$ be as in Lemma~\ref{lem:HScont} and $g$ be a solution of~(\ref{IsoCond}). Then for $s\in(0,s_0]$, there exists an ergodic purely generated MPS on the spin-$1$ chain corresponding to the matrices $\{w_i(s)\}_{i=\pm,0}$.
\end{lemma}
\begin{proof}
First, we note that Condition~(\ref{IsoCond}) ensures that $W(s)$ is an isometry, or equivalently that $\idtyty$ is an eigenvector of $\widehat{\bbT}(s)$ for the eigenvalue $1$, as
\begin{equation*}
\widehat{\bbT} (s;\idtyty) = \sum_{i=\pm,0} w_i(s)\str w_i(s) = \idtyty\,.
\end{equation*}
In fact, all right eigenvectors are given by
\begin{equation*}
R_1(s)=\idtyty\,,\qquad
R_2(s)=\sigma^+\,, \qquad 
R_3(s)=\sigma^-\,, \qquad
R_4(s)=\varrho_2(s) P - \varrho_1(s)Q
\end{equation*}
where $\varrho_2(s) = \left(\sin^2(s) / g^2(s)\right) \varrho_1(s)$. The associated eigenvalues are given by $t_1(s)=1$, $t_2(s)=t_3(s)=-f(s)\cos^2(s)$ and $t_4(s)=\cos^2(s)-g(s)^2 = f(s)^2\cos^2(s)-\sin^2(s)$. The assumptions on $f(s)$ ensure that $1$ is non degenerate for all $s\in[0,s_0]$. The corresponding unique left eigenvector is given by $\varrho(s) = \varrho_1(s) P + \varrho_2(s)Q$. With the normalization, we get
\begin{equation*}
\varrho_1(s) = \frac{g(s)^2}{g(s)^2 + \sin^2(s)}\,,\qquad \varrho_2(s) = \frac{\sin^2(s)}{g(s)^2 + \sin^2(s)} = 1-\varrho_1(s)\,,
\end{equation*}
and the matrix $\varrho$ defines a faithful state on the auxiliary algebra $\caM_2$ for $s>0$. Hence, the triple $(\bbT,\idtyty,\varrho)$ generates a translation invariant state on the chain algebra, which is moreover extremal since the peripheral spectrum of $\widehat{\bbT}(s)$ is trivial for all $s\in[0,s_0]$.
\end{proof}
In fact, the transfer operator $\widehat{\bbT}(s)$ can be completely diagonalized
\begin{equation*} %\label{EDiag}
\widehat{\bbT}(s;A) = \sum_{i=1}^4 t_i(s)\, \Tr(L_i(s) A) R_i(s) \,,
\end{equation*}
with the left eigenvectors $L_i(s)$ given by
\begin{equation*}
L_1(s)=\varrho(s)\,,\qquad 
L_2(s)=R_3(s) \qquad L_3(s)=R_2(s)\,,\qquad
L_4(s)=P-Q = \sigma^z \,.
\end{equation*}
The normalization has been chosen so that $\Tr(L_iR_i) = 1$.

The matrices $\{w_i(s)\}_{i=0}^n$ form a representation of the following deformation of the algebra~(\ref{AKLTCommutation}) of the AKLT model into that of the PVBS model:
\begin{gather}\label{sCommutation1}
f(s)w_0(s)w_-(s) = -w_-(s)w_0(s)\,,\qquad w_0(s)w_+(s) = -f(s)w_+(s)w_0(s)\,,\qquad w_\pm(s)^2 = 0\,,Ê\\
w_-(s)w_+(s) + f(s)^2 w_+(s)w_-(s) = -\frac{g(s)\sin(s)}{\cos^2(s)}w_0(s)^2\,.\label{sCommutation2}
\end{gather}
Although the representation~(\ref{wMatrices}) holds for all $s$, including $s=0$, it fails to be faithful at the boundary point as $w_+(0) = 0$. A non-trivial four-dimensional representation of the algebra at that point was given in Section~\ref{Sec:PVBS}. Interestingly, the following lemma shows that the discontinuity in the dimension of the minimal non-trivial representation of the smooth deformation~(\ref{sCommutation1},\ref{sCommutation2}) is necessary.
\begin{lemma}
Let $v_-, v_0, v_+$ be a two-dimensional representation of the commutation relations~(\ref{cr1},\ref{cr2}) with $n=2$ and $\lambda_- \neq \lambda_+$. Then $v_-=0$ or $v_+=0$.
\end{lemma}
\begin{proof}
Assume that $v_-\neq0$ and $v_+\neq0$. The condition $\lambda_- \neq \lambda_+$ together with the commutation relations for $(-,0)$ and $(+,0)$ imply that $v_- \neq v_+$. In two dimensions, any nonzero nilpotent matrix must be unitarily equivalent to $\sigma^+$. Hence, we let $v_- = \alpha\sigma^+$, $v_+ = \beta U\str v_- U$, with $U\neq\idtyty$ and $\alpha\neq 0$, $\beta\neq 0$. Taking the trace of the relation for $(-,+)$ and using cyclicity yields $\Tr(v_-v_+) = 0$. This implies $(v_+)_{21}=0$, so that $\overline{U_{12}}U_{21} = 0$. Hence, by unitarity, $U=\idtyty$, which is a contradiction.
\end{proof}

At this point, one may also wonder why the commutation of $w_-(s)$ with $w_0(s)$ and that of $w_+(s)$ with $w_0(s)$ involve parameters that are inverses of each other. Indeed, at the PVBS point, $\lambda_+$ and $\lambda_-$ can be chosen independently of each other. Here again, it turns out that it is in fact dictated by the form of the commutation relations themselves.
\begin{lemma}\label{lma:CommutationSymmetry}
We consider the following quadratic relations:
\begin{gather*}
w_-w_0 = C_{-0}w_0 w_-\,,\qquad w_+w_0 = C_{+0}w_0 w_+ \,,\qquad w_\pm^2 = 0\,,Ê\\
w_-w_+ = C_{-+}w_+ w_- + D w_0^2\,,
\end{gather*}
where all coefficients $C_{ij}$ and $D$ are non zero. If there exists a faithful representation such that $w_0^3 \neq 0$, then $C_{+0} = 1/C_{-0}$. In particular,
\begin{equation*}
\left[w_+ w_-, w_i \right] = 0\,,\qquad i=0, \pm\,.
\end{equation*}
\end{lemma}
\begin{proof}
On the one hand,
\begin{equation*}
w_-w_+w_0 = C_{-0}C_{+0}w_0w_-w_+ = C_{-0}C_{+0}Dw_0^3 + C_{-0}C_{+0}C_{-+}w_0w_+w_-\,,
\end{equation*}
while on the other hand,
\begin{equation*}
w_-w_+w_0 = C_{-+}w_+w_-w_0 + Dw_0^3 = C_{-+}C_{-0}C_{+0} w_0w_+w_- + Dw_0^3
\,,
\end{equation*}
so that $C_{-0}C_{+0}Dw_0^3 = Dw_0^3$ which yields the claim.
\end{proof}
In the PBVS case, $D=0$ and the lemma does not apply. 

The condition $C_{+0} = 1/C_{-0}$ is equivalent to the linear independence of $\{w^I:I\in S^{(3)}\}$ at the level of the abstract algebra, and ensures that the intersection property holds.
\begin{lemma}\label{lem:IntersectionPath}
For $0<s\leq s_0$, let $\caR_N = \mathrm{Ran}(\Upsilon_N)$. For $N\geq 2$, the spaces $\caR_N(s)$ satisfy the intersection property~(\ref{Intersection}). Moreover, $\mathrm{dim}(\caR_N(s)) = 4$.
\end{lemma}
\begin{proof}
The intersection property follows from Lemma~\ref{lma:IntProp}. In the notation thereof, we choose $S^{(2)} = \{(00), (0+), (0-), (+-)\}$, so that
\begin{equation*}
S^{(N)} = \{(0\cdots0), (0\cdots0+), (0\cdots0-), (0\cdots0+-)\}
\end{equation*}
and the linear independence of the corresponding $\{w^I:I\in S^{(N)}\}$ is immediate from the explicit representation of the matrices~(\ref{wMatrices}). That $\mathrm{dim}(\caR_N(s)) = 4$ follows from $\abs{S^{(N)}} = 4$ and again the linear independence.
\end{proof}

It will be useful to have an explicit description of a well-behaved basis of $\caG_N(s)$. For $2\leq N < \infty$ and $s\in(0,s_0]$ let
\begin{equation} \label{zetas}
\zeta^{\mu}_N(s) := \Upsilon_N(s;A^\mu_N(s))\,,
\end{equation}
where
\begin{equation*}
A^{1}_N(s) = P + q_N(s)Q\,,\quad
A^{2}_N(s) = \sigma^-\,,\quad
A^{3}_N(s) = \frac{g(s)}{\sin(s)}\sigma^+ \,,\quad
A^{4}_N(s) = \frac{g(s)}{\sin(s)}\left[(-f(s))^N P - Q\right]\,,
\end{equation*}
with 
\begin{equation} \label{q}
q_N(s) = \frac{t_2(s)^N-(-f(s))^N\left(\varrho_1(s)+\varrho_2(s)t_4(s)^N\right)}{t_2(s)^N(-f(s))^N-\left(\varrho_2(s)+\varrho_1(s)t_4(s)^N\right)}\,.
\end{equation}
Now, (\ref{overlap}) reads
\begin{equation}
\braket{\zeta^{\mu}_N(s)}{\zeta^{\nu}_N(s)}
= \sum_{j=1}^4 t_j(s)^N \,\Tr\left(L_j(s)A^{\mu*}_N(s)R_j(s)A^\nu_N(s)\right)\,, \label{overlap2}
\end{equation}
for any $1\leq \mu,\nu\leq 4$. With $R_j(s), L_j(s)$ given above, this vanishes for all pairs $\mu<\nu$, the case $\mu=1,\nu=4$ depending on 
%%
%\begin{multline*}
%\sin(s)\sum_{j=1}^4 t_j(s)^N \,\Tr\left(A^{1*}_N(s)R_j(s)A^4_N(s)L_j(s)\right) = \\
%(-f(s))^N\left(\varrho_1(s) + \varrho_2(s) t_4(s)^N\right) - q_N(s)\left(\varrho_2(s) + \varrho_1(s) t_4(s)^N \right) + \left(q_N(s)(-f(s))^N-1\right) t_2(s)^N \,,
%\end{multline*}
%%
the choice~(\ref{q}) for $q_N(s)$. Moreover when $\mu=\nu$, (\ref{overlap2}) yields
\begin{align}
\left\Vert\zeta_N^1(s)\right\Vert^2&= \left(\varrho_1(s) + \varrho_2(s) t_4(s)^N\right) + 2q_N(s)t_2(s)^N + q_N(s)^2 \left(\varrho_2(s) + \varrho_1(s) t_4(s)^N\right) \label{Norm1}\\
\left\Vert\zeta_N^2(s)\right\Vert^2&= \left\Vert\zeta_N^3(s)\right\Vert^2 = \varrho_1(s)\left(1 - t_4(s)^N\right) \label{Norm2}\\
%\left\Vert\zeta_N^3(s)\right\Vert^2&= \frac{g(s)^2}{\sin^2(s)}\varrho_2(s)\left(1 - t_4(s)^N\right) \label{Norm3}\\
\left\Vert\zeta_N^4(s)\right\Vert^2&= \frac{g(s)^2}{\sin^2(s)}\bigg[\left(\varrho_2(s) + \varrho_1(s) t_4(s)^N\right) \label{Norm4} \\
&\phantom{= \frac{g(s)^2}{\sin^2(s)}\big[} + f(s)^{2N}\left(\varrho_1(s) + \varrho_2(s) t_4(s)^N\right) - 2(-f(s))^N t_2(s)^N \bigg] \nonumber\,.
\end{align}
These norms remain strictly positive for all $s > 0$, and the set $(\zeta_N^\mu(s))_{\mu=1}^4$ forms an well-defined orthogonal set spanning the four-dimensional space $\caG_N(s)$ for $s>0$. In fact, $\lim_{s\to0}q_N(s)$ exists and the four vectors $\zeta_N^\mu(s)$ can be extended by continuity to $s=0$. Furthermore,
\begin{equation*}
\lim_{N\to\infty} \lim_{s\to 0} \left\Vert\zeta_N^\mu(s)\right\Vert = \lim_{s\to 0}\lim_{N\to\infty} \left\Vert\zeta_N^\mu(s)\right\Vert = \lim_{s\to 0} \varrho_1(s) = 1
\end{equation*}
for all $\mu$. From now on, we define
\begin{equation} \label{RN0}
\caR_N(0) := \mathrm{span}\{\zeta_N^\mu(0^+):\mu=1,\ldots,4\}\,,
\end{equation}
and in particular,
\begin{equation*}
\caR_N(0) \neq \mathrm{Ran}(\Upsilon_N(0))\,.
\end{equation*}
Indeed, since $w_+(0) = 0$, the range of $\Upsilon_N(0)$ is two-dimensional. The map $\Upsilon_N(0)$ generates the ground states of the system on the infinite chain without boundary and the chain with a left boundary, but not on the other half-infinite system. With minor modifications, the roles of left and right could be exchanged.

As we already noted in the introduction, this example emphasizes the importance of constructing a family of Hamiltonians and ground states, as a smooth family of matrix product maps $\Upsilon_N(s)$ does not necessarily generate the desired continuity if the parent Hamiltonians are simply constructed pointwise.

Despite being a simple corollary of the above discussion, the following proposition is essential: it identifies the MPS constructed from the path of matrices for $s>0$ and their limits at $s=0^+$ with the ground states of the Hamiltonians of Lemma~\ref{lem:HScont}.
\begin{prop}
Let $\caR_N(s) = \mathrm{Ran}(\Upsilon_N(s))$ for $s>0$ and $\caR_N(0)$ be defined by~(\ref{RN0}). If $H_{[a,b]}(s)$ are the Hamiltonians of Lemma~\ref{lem:HScont}, with $N=b-a+1$, then
\begin{equation*}
\caR_N(s) = \mathrm{Ker}\left(H_{[a,b]}(s)\right)\,,
\end{equation*}
for all $s\in[0,s_0]$.
\end{prop}
\begin{proof}
By construction, and for all $s\geq 0$,
\begin{equation} \label{sGS2}
\begin{split}
\zeta^{1}_2(s) &= - f(s)^2 g(s) \sin (s) e_-\otimes e_+ + \cos^2(s)(1+f(s)^4)e_0\otimes e_0 - g(s)\sin(s) e_+\otimes e_-\,, \\
\zeta^{2}_2(s) &= -f(s)e_0\otimes e_- + e_-\otimes e_0\,, \\
\zeta^{3}_2(s) &= - f(s)e_+\otimes e_0 + e_0\otimes e_+\,, \\
\zeta^{4}_2(s) &= e_-\otimes e_+ - f(s)^2e_+\otimes e_- \,,
\end{split}
\end{equation}
where we used that $q_2(s) = f(s)^2$. These span the orthogonal complement of the space spanned by the five vectors of~(\ref{G2}), so that $\caR_2(s) = \mathrm{Ker}(h(s)) = \mathrm{Ker}(H_{[x,x+1]}(s))$. The statement for all $N\geq2$ follows by the intersection property.
\end{proof}
Interestingly, we see that the path chosen here maps the spin-$0$ singlet ground state of the AKLT model to the product state of the PVBS model, while the spin-$1$ triplet is deformed to the edge states. 

The last ingredient needed in the proof of Theorem~\ref{thm:AKLT_phase} is a uniform lower bound on the spectral gap along the path of Hamiltonians. The following lemma, whose proof is very reminiscent of that of Proposition~\ref{Prop:PVBS_Gap}, is the key. Recall that $G_N(s)$ is the orthogonal projection on $\caG_N(s)$.
\begin{lemma} \label{lma:Commutation}
Consider the following family of functions on $s\in[0,s_0]$
\begin{equation*}
g_{k,N}(s):=\left \Vert G_{[N-k+2,N+1]}(s)\left(G_{[1,N]}(s)-G_{[1,N+1]}(s)\right) \right\Vert\,.
\end{equation*}
for $k\in\bbN$ and $N\geq k$. Then:
\begin{enumerate}
\item For all $N \geq 2k-2$ and for all $s\in[0,s_0]$, there is a $C(s)$ such that
\begin{equation} \label{ExpDecay}
g_{k,N}(s)^2 \leq C(s)\varepsilon^{k-1}
\end{equation}
where $\varepsilon = \sup_{s\in[0,s_0]}\left(\max_{i\in\{2,3,4\}}\left\vert t_i(s) \right\vert\right)$. Moreover, $s\mapsto C(s)$ is continuous on $(0,s_0]$.
\item Let $0<f<1$ be fixed. Let $0\leq s_1,s_2\leq \pi/2$ be defined by
\begin{equation*}
\sin^2(s_1)=\frac{1-f^2}{1+f^2}\,,\qquad \sin^2(s_2)=\frac{f^2(1-2\ln f)}{1+f^{2}}
\end{equation*}
and let $\delta = \min\{s_1,s_2,s_0/2\}$. Assume that $f(s) = f$ for $s<\delta$. There exists $k_1$ such that $k\geq k_1$ implies $g_{k,N}(s)^2<1/k$ for all $s<\delta$ and $N \geq 2k-2$.
\end{enumerate}
\end{lemma}
\begin{proof}
\emph{i. Pointwise bound.} By a direct translation of the steps leading to~(\ref{PVBS:start}), but with $s$-dependent objects and ground state vectors indexed by $\mu=1,\ldots,4$, we can write
\begin{equation}\label{Norm:Start}
g_{k,N}(s)^2 = \sup_{\Phi\in \caK(s)} \sum_{\mu,\nu}\sum_{\tau,i}\frac{\overline{\phi_{\mu,\nu;k,N}(s)}\varphi_{\tau,i;N}(s)}{\Vert \zeta^\mu_{N-k+1}(s) \Vert \Vert \zeta^\nu_{k}(s) \Vert \Vert \zeta^\tau_{N}(s) \Vert}\braket{\zeta^{\mu}_{N-k+1}(s)\otimes\zeta^{\nu}_{k}(s)}{\zeta^\tau_N(s)\otimes e_i}
\end{equation}
where 
\begin{equation*}
\caK(s) := \left\{\Phi\in\caH_{[1,N+1]}:\Vert\Phi\Vert = 1\text{ and }\Phi\in\caG_{[1,N]}(s)\cap\caG_{[1,N+1]}(s)^\perp\right\}\,,
\end{equation*}
$\phi_{\mu,\nu;k,N}(s) = \langle\hat{\zeta}^{\mu}_{N-k+1}(s)\otimes\hat{\zeta}^{\nu}_{k}(s),\Phi\rangle$ and $\varphi_{\tau,i;N}(s) = \langle\hat{\zeta}^{\tau}_{N}(s)\otimes e_i,\Phi\rangle$. The normalization $\Vert \Phi \Vert = 1$ and the orthogonality of the vectors $(\zeta^\tau_N(s))_{\tau=1}^4$ imply that $\vert \phi_{\mu,\nu;k,N}(s) \vert \leq 1$ as well as $\vert \varphi_{\tau,i;N}(s) \vert \leq 1$.

If $s>0$, scalar products of ground states can be expressed using the transfer operator and its spectral decomposition, namely
\begin{align*}
\braket{\zeta^{\mu}_{N-k+1}\otimes\zeta^{\nu}_{k}}{\zeta^\tau_N\otimes e_i} &= 
\sum_{p,q=1}^4t_p^{k-1}t_q^{N-k+1}\Tr\left(R_qA_N^\tau L_pw_i\str(A_k^\nu)\str R_pL_q(A_{N-k+1}^\mu)\str\right)\,, \\
\braket{\zeta^\nu_{N+1}}{\zeta^\tau_N\otimes e_i} &= \sum_{p=1}^4t_p^{N}\Tr\left(R_pA_{N}^\tau L_pw_i\str(A_{N+1}^\nu)\str\right)\,.
\end{align*}
With the second identity, the condition $\Phi \in \caG_{[1,N+1]}(s)^\perp$ reads
\begin{equation} \label{Cond:OrthDecay}
\sum_{\tau,i}\frac{\varphi_{\tau,i;N}}{\Vert \zeta^\tau_N \Vert}\Tr\left(A_{N}^\tau \rho w_i\str(A_{N+1}^\nu)\str\right) = \sum_{\tau,i}\frac{\varphi_{\tau,i;N}}{\Vert \zeta^\tau_N \Vert}\sum_{p=2}^4t_p^{N}\Tr\left(R_pA_{N}^\tau L_pw_i\str(A_{N+1}^\nu)\str\right)\,.
\end{equation}
The first identity on the other hand yields
\begin{align*}
\frac{\braket{\zeta^\mu_{N-k+1}\otimes\zeta^\nu_{k}}{\Phi} }{\Vert \zeta^\mu_{N-k+1} \Vert \Vert \zeta^\nu_{k} \Vert} 
&= \frac{1}{\Vert \zeta^{\mu}_{N-k+1} \Vert \Vert \zeta^{\nu}_{k} \Vert} \Bigg[\sum_{\tau,i}\frac{\varphi_{\tau,i;N}}{\Vert \zeta^\tau_{N} \Vert} \Tr\left(A_{N}^\tau \rho w_i\str(A_k^\nu)\str \rho (A_{N-k+1}^\mu)\str\right) \\
&\quad + \sum_{p+q>2} t_p^{k-1} t_q^{N-k+1} \sum_{\tau,i}\frac{\varphi_{\tau,i;N}}{\Vert \zeta^\tau_{N} \Vert}\Tr\left(R_qA_{N}^\tau L_pw_i\str(A_k^\nu)\str R_pL_q(A_{N-k+1}^\mu)\str\right)\Bigg]\,.
\end{align*}
Since the matrices $(A_{N+1}^\nu)_{\nu=1}^4$ are linearly independent, there exist coefficients $c_{\phantom{\nu}{\phantom\mu}\sigma}^{\nu\mu}(k,N)$ such that
\begin{equation*}
(A_k^\nu)\str \rho (A_{N-k+1}^\mu)\str = \sum_{\sigma=1}^4c_{\phantom{\nu}{\phantom\mu}\sigma}^{\nu\mu}(k,N)(A_{N+1}^\sigma)\str\,.
\end{equation*}
Note that $c_{\phantom{\nu}{\phantom\mu}\sigma}^{\nu\mu}(k,N)$ are uniformly bounded in $k,N$. It remains to use this decomposition and the orthogonality condition~(\ref{Cond:OrthDecay}) in the first line above to obtain
\begin{multline} \label{g}
g_{k,N}^2 = \sup_{\Phi:\Vert \Phi \Vert = 1} \sum_{\mu,\nu}\sum_{\tau,i}\frac{\overline{\phi_{\mu,\nu;k,N}}\varphi_{\tau,i;N}}{\Vert \zeta^\mu_{N-k+1} \Vert \Vert \zeta^\nu_{k} \Vert \Vert \zeta^\tau_{N} \Vert} \Bigg[
\sum_{\sigma} c_{\phantom{\nu}{\phantom\mu}\sigma}^{\nu\mu}(k,N)
\sum_{p=2}^4t_p^{N}\Tr\left(A_{N}^\tau \rho w_i\str(A_{N+1}^\sigma)\str\right) \\
+\sum_{p+q>2} t_p^{k-1}t_q^{N-k+1} 
\Tr\left(R_qA_{N}^\tau L_pw_i\str(A_k^\nu)\str R_pL_q(A_{N-k+1}^\mu)\str\right)\Bigg]
\end{multline}
Recalling that there exist $0<C_m\leq C_M<\infty$ such that
\begin{equation*}
C_m\leq \Vert \zeta_N^\mu(s)\Vert \leq C_M
\end{equation*}
uniformly in $N$ and noting that the traces are finite for any fixed $s>0$, the fact that the peripheral spectrum of the transfer operator is empty implies the exponential decay~(\ref{ExpDecay}) of $g_{k,N}(s)^2$ for all $N\geq2k-2$. Finally, we note that the case $s=0$ corresponds to the PVBS model and has been treated independently in the proof of Proposition~\ref{Prop:PVBS_Gap}. The continuity of $C(s)$ follows from that of all elements of~(\ref{g}) for $s>0$.

\emph{ii. Uniform bound.} A priori, the function $C(s)$ has a singularity at $s=0$ due to the factors $\sin(s)^{-1}$ of the matrices $A^3_\cdot$ and $A^4_\cdot$, appearing both in the traces and in the coefficients $c_{\phantom{\nu}{\phantom\mu}\sigma}^{\nu\mu}(k,N)$. If $B_{k,N}(s)$ is the square bracket in~(\ref{g}), we write
\begin{equation*}
B_{k,N}(s) = \frac{c_{k,N}(s)}{\sin^3(s)} + r_{k,N}(s)
\end{equation*}
by gathering all singular contributions in the first term. The regular part $r_{k,N}(s)$ containing no negative power of $\sin(s)$, it can immediately be bounded above as
\begin{equation*}
\abs{r_{k,N}(s)}\leq C_r\varepsilon^{k-1}\,,
\end{equation*}
with an $s$-independent constant $C_r$, for all $N\geq2k-2$. For $s<s_1$, and recalling that $f(s)$ is constant there, $\varrho_1(s)$ can be expanded in powers of $\sin(s)$
\begin{equation*}
\varrho_1(s) = 1 - \frac{1}{1-f^2}\sum_{j=1}^\infty\left(-\frac{1+f^2}{1-f^2}\right)^{j-1}\sin^{2j}(s)\,,
\end{equation*}
and similarly for $\varrho_2(s) = 1-\varrho_1(s)$. Since the coefficients of a product are given by convolution, the expansions of $\varrho_1(s)^2, \varrho_2(s)^2$ and $\varrho_1(s)\varrho_2(s)$ converge absolutely for $s<s_1$. Hence, the singular part of~(\ref{g}) has the following convergent representation,
\begin{equation}\label{c_Expansion}
\frac{c_{k,N}(s)}{\sin^3(s)} = \sum_{j=0}^\infty \Bigg(\sum_{p=2}^4 t_p(s)^Na_j + \sum_{p+q>2}t_p(s)^{k-1} t_q(s)^{N-k+1}b_j \Bigg)
 \sin^{j-3}(s)\,,
\end{equation}
where the coefficients $a_j = a_j(k,N)$ and $b_j = b_j(k,N)$ are independent of $s$ and uniformly bounded in $k,N$. The terms $j\geq 3$ are regular and summable for $s<s_1$ and can be estimated again by a constant times $\varepsilon^{k-1}$. 

Now, we use the binomial expansions of $t_{p}(s)^n$ in the $0\leq j \leq 2$ terms of~(\ref{c_Expansion}) and reorganize the sums to obtain
\begin{multline*}
\sum_{i=0}^2 \sum_{j=0}^i \Bigg(
 a_j \sum_{p=2}^4 \sum_{\alpha:j+2\alpha = i} \tau^N_{\phantom{^N}p,\alpha} 
 + b_j \sum_{p+q>2} \sum_{\substack{\alpha,\beta : \\ j+2\alpha + 2\beta = i}}\tau^{k-1}_{\phantom{^{k-1}}p,\alpha} \tau^{N-k+1}_{\phantom{^{N-k+1}}p,\beta} 
\Bigg)\sin^{i-3}(s) \\
+ \sum_{j=0}^2 \Bigg(
a_j \sum_{p=2}^4 \sum_{\alpha:j+2\alpha\geq 3}^N\tau^N_{\phantom{^N}p,\alpha} \sin^{j+2\alpha-3}(s) 
+ b_j \sum_{p+q\geq 1} \sum_{\substack{\alpha,\beta\\ j+2\alpha + 2\beta \geq 3}}^{k-1,N-k+1} \tau^{k-1}_{\phantom{^{k-1}}p,\alpha} \tau^{N-k+1}_{\phantom{^{N-k+1}}p,\beta} \sin^{j+2\alpha+2\beta-3}(s)
\Bigg)\,,
\end{multline*}
where
\begin{equation*}
\tau^n_{\phantom{n}2,j} = \tau^n_{\phantom{n}3,j} = (-f)^n\binom{n}{j}(-1)^j\,,\qquad\tau^n_{\phantom{n}4,j} = f^{2n}\binom{n}{j}\left(-(1+f^{-2})\right)^j\,.
\end{equation*}

Now, the a priori bound $g_{k,N}(s)^2\leq 1$ implies that $c_{k,N}(s) = \caO_{k,N}(\sin^3(s))$, and in turn that the first line vanishes. In the second line, where all exponents of $\sin(s)$ are non-negative, it suffices to observe that
\begin{equation*}
\sum_{k=\gamma}^n\binom{n}{k} x^{k-\gamma} \leq \sum_{j=0}^{n-\gamma}\frac{n^{j+\gamma}}{(j+\gamma)!}\abs{x}^j\leq n^\gamma \ep{n\abs{x}}\,.
\end{equation*}
Hence,
\begin{equation*}
\Bigg\vert\sum_{\alpha:j+2\alpha\geq 3}^N\tau^N_{\phantom{^N}4,\alpha} \sin^{j+2\alpha-3}(s) \Bigg\vert
\leq f^{2N} \left((1+f^{-2})N\right)^{\lceil(3-j)/2\rceil}\ep{N (1+f^{-2})\sin^2(s)}
\end{equation*}
and similarly for the other terms. As $N\to\infty$, the left hand side decays exponentially for $s<s_2$. Summarizing, all contributions to~(\ref{c_Expansion}) show a uniform exponential decay in $k$ for $N\geq2k-2$, so that for $k$ large enough, $g_{k,N}(s)^2<1/k$ for all $s<\delta$.
\end{proof}
Note that the specific assumptions sufficient for point (ii) to hold are of technical nature and we have not aimed for the optimal conditions. Any neighborhood of $s=0$ in which the gap remains open will do, as part (i) ensures the existence of a gap for $s>0$. Similarly, the choice of a constant $f$ around the origin simplifies the perturbative calculations in the proof, but is not necessary. A general theory of the continuity of the spectral gap for such families of Hamiltonians still needs to be developed.
\begin{proof}[Proof of Theorem~\ref{thm:AKLT_phase}]
Lemma~\ref{lem:HScont} introduces a smooth path of Hamiltonians for $s\in[0,s_0]$. The continuity of the map $s\mapsto C(s)$ on $(0,s_0]$ introduced in Lemma~\ref{lma:Commutation}(i) implies that it is bounded on $[\delta,s_0]$. Therefore, there exists $k_0$ such that $k\geq k_0$ implies
\begin{equation*}
g_{k,N}(s)^2\leq C(s) \varepsilon^{k-1}<1/k\qquad s\geq \delta\,,
\end{equation*}
for all $N\geq 2k-2$. Together with (ii) of the same lemma, $g_{k,N}(s)^2<1/k$ for all $k\geq \max\{k_0,k_1\}$ and all $s\in[0,s_0]$. Hence, the martingale result~(\ref{GapMartingale}) implies that the path is gapped,
\begin{equation*}
\inf_{s\in[0,s_0], N\in\bbN}\gamma(s,N)>0\,.
\end{equation*}
Hence, by~\cite{BMNS}, there exist quasi-local automorphisms $\alpha^\Gamma$, for $\Gamma=\bbZ,(-\infty,0],[1,+\infty)$ on the algebras of the infinite chain, of the right infinite chain with a left edge and of the left infinite chain with a right edge, and concludes the proof of the theorem.
\end{proof}
%

%%%%%%%%%%%%%%%%%%%%%%%%%%%%%%%%%%%%%%%%%%%%%%%%%%%%%%%%%%%%%%%%%%%%%

\section{The $SO(2J+1)$ models}

%%%%%%%%%%%%%%%%%%%%%%%%%%%%%%%%%%%%%%%%%%%%%%%%%%%%%%%%%%%%%%%%%%%%%

As a further example of the classification of gapped phases using the PVBS classes, we consider the $SO(d)$ invariant models introduced in~\cite{PhysRevB.78.094404}, which are one possible family of generalization of the AKLT model. The local Hilbert space is $d$ dimensional and carries the fundamental representation $\caD$ of $SO(d)$. In order to define a nearest neighbor interaction, we consider the irreducible decomposition of the tensor product of two fundamental representations into
\begin{equation} \label{Decomp}
\caD \otimes \caD = 1 \oplus \caA \oplus \caS
\end{equation}
where $1$ is the invariant singlet, $\caS$ contains all other symmetric vectors, and $\caA$ is the set of the antisymmetric vectors. Let $P_\caS$ be the orthogonal projection onto $\caS$, and let the Hamiltonian be
\begin{equation}\label{SOHam}
H_{[a,b]} = \sum_{x=a}^{b-1}(P_\caS)_{x,x+1}\,.
\end{equation}
By construction it is invariant under the action of $SO(d)$. In the case $d=3$, the decomposition~(\ref{Decomp}) coincides with the Clebsch-Gordan decomposition of the tensor product of the tensor product of two spin-$1$ representations of $\mathfrak{su}(2)$, and $P_\caS$ is the AKLT interaction. In the sequel, we shall focus on the odd case, $d=2J+1$.

The matrix product representation of the $SO(2J+1)$ models' ground states can be chosen to be the generators $\{Z_\alpha = Z_\alpha^*|\alpha=0,\ldots,2J\}$ of the Clifford algebra $\caC_{2J+1}$ satisfying the quadratic relations
\begin{equation}\label{Clifford}
Z_\alpha Z_\beta + Z_\beta Z_\alpha = 2\delta_{\alpha\beta} \idtyty.
\end{equation}
The fact that the symmetry remains unbroken in the ground states is manifest in the invariance of~(\ref{Clifford}) under the transformation $Z'_\gamma = \sum_{\alpha} O_{\gamma\alpha} Z_\alpha$ for any $O\in SO(2J+1)$. The correct normalization is obtained by setting $\tilde{Z}_\alpha = \gamma Z_\alpha$ and imposing the isometry condition $\idtyty = \sum_\alpha  \tilde{Z}_\alpha\str \tilde{Z}_\alpha$. Since $Z_\alpha^2 = \idtyty$ for all $\alpha$, we have that 
\begin{equation} \label{CliffordNormalization}
\abs{\gamma}^2 = \frac{1}{2J+1}\,.
\end{equation}

A representation of $\caC_{2J+1}$ can conveniently be obtained from a representation of the CAR algebra $\caA_-(\bbC^{J})$. $\caA_-(\bbC)$ is generated by $\idtyty$, $\sigma^+$ and $\sigma^-$, and the higher dimensional fermionic creation and annihilation operators are given by
\begin{equation} \label{CAR}
a\str_j = \underbrace{(-1)^Q\otimes \cdots \otimes (-1)^Q}_{j-1} \otimes \sigma^+ \otimes \underbrace{\idtyty \otimes \cdots \otimes \idtyty}_{J-j} \in \caA_-(\bbC^{J})\,.
\end{equation}
Note that $(-1)^Q = \sigma_z$. Then, 
\begin{equation} \label{Z0}
Z_0 = \prod_{j=1}^J (2a\str_j a_j - 1) = \bigotimes_{j=1}^J(-1)^Q\,,
\end{equation}
and for $1\leq j \leq J$,
\begin{equation*}
Z_{2j-1} = a_j + a\str_j\,,\qquad Z_{2j} = \iu(a_j - a\str_j)\,.
\end{equation*}
The canonical anticommutation relations of the $\{a\str_j, a_j\}_{j=1}^J$ imply the Clifford relations for the $\{Z_\alpha\}_{\alpha=0}^{2J}$. This representation clearly is of dimension $2^J$.
\begin{lemma}
The ground states of~(\ref{SOHam}) are MPS generated by the matrices $\{V_\alpha\}_{\alpha=0}^{2J}$, where 
\begin{equation*}
{V}_0 = -\frac{1}{\sqrt{2J+1}}\,Z_0\,,\qquad {V}_{2j-1} = \sqrt{\frac{2}{2J+1}}\,a_j\,,\qquad {V}_{2j} = -\sqrt{\frac{2}{2J+1}}\,a_j\str \,,
\end{equation*}
for $1\leq j\leq J$, and the operators $a_j^\sharp$ are given by (\ref{CAR}) and $Z_0$ by (\ref{Z0}).
\end{lemma}
\begin{proof}
We prove that the matrices of the lemma generate the same vectors as $\{\gamma Z_\alpha\}_{\alpha=0}^{2J}$. Quite generally, the pullback of a change of local basis on the chain $\psi_\alpha\mapsto \tilde{\psi}_\beta = \sum_\alpha U_{\beta\alpha}\psi_\alpha$ to the generating matrices $v_\alpha$ of a MPS is given by $v_\alpha\mapsto \tilde{v}_{\beta} = \sum_\alpha \overline{U_{\beta\alpha}}v_\alpha$. Indeed,
\begin{equation*}
\sum_\beta \tilde{\psi}_\beta\otimes \tilde{v}_{\beta} \chi = \sum_{\alpha,\beta} U_{\beta\alpha}\psi_\alpha \otimes \tilde{v}_{\beta} \chi = \sum_{\alpha} \psi_\alpha \otimes v_\alpha \chi
\end{equation*}
implies $v_\alpha = \sum_{\beta} (U^T)_{\alpha\beta} \tilde{v}_{\beta}$. The definitions
\begin{equation*}
V_0 = \gamma Z_0\,,\qquad{V}_{2j-1} = \alpha_ja_j\qquad {V}_{2j} = \beta_ja_j\str
\end{equation*}
correspond to a change of basis given by a matrix $U$ with $U_{00} = 1$ and $2\times2$ blocks  
\begin{equation*}
\frac{1}{2\gamma} \begin{pmatrix}
\alpha_j & -\iu\alpha_j \\ \beta_j & \iu\beta_j
\end{pmatrix}
\end{equation*}
for $1\leq j\leq J$. Unitarity imposes that $\abs{\alpha_j}^2 = \abs{\beta_j}^2 = 2\abs{\gamma}^2$. The claim follows by choosing $\gamma = -(2J+1)^{-1/2}$ in accordance with~(\ref{CliffordNormalization}).
\end{proof}
The arbitrary choice of sign is designed to match exactly~(\ref{AKLTws}) in the case $J=1$. With this mapping from the Clifford algebra to the CAR algebra, it is now immediate to generalize the path of matrix product maps and therefore of Hamiltonians from $J=1$ to an arbitrary $J$, thereby introducing a smooth path of Hamiltonians between any $SO(2J+1)$ model and the PVBS model with $J$ particles of both types. If the path is gapped indeed, this places these higher spin models into the PVBS classification.

We start by introducing a deformation of the CAR algebra to relate it to the PVBS algebra with $n_L = n_R = J$. For any complex number $\lambda$,
\begin{gather}
\label{twistedCAR_1}
\sigma^+\sigma^- + \lambda \sigma^-\sigma^+ = \lambda^{Q}\,,\\
\lambda \cdot \sigma^- \lambda^Q =  \lambda^Q \sigma^- \,,\qquad \sigma^+ \lambda^Q = \lambda \cdot \lambda^Q \sigma^+\,.
\label{twistedCAR_2}
\end{gather}
Given a vector $\Lambda=(\lambda_1,\ldots,\lambda_J)$, we define the twisted creation operators $(a\str_j(\Lambda))_{j=1}^J$ by
\begin{equation*}
a\str_j(\Lambda) = (-\lambda_1)^Q\otimes \cdots \otimes (-\lambda_{j-1})^Q \otimes \sigma^+ \otimes \lambda_{j+1}^Q \otimes \cdots \otimes \lambda_{J}^Q\,,
\end{equation*}
and
\begin{equation*}
a_0(\Lambda) = (-\lambda_1)^Q \otimes \cdots \otimes (-\lambda_J)^Q \,.
\end{equation*}
Clearly, $(a\str_j(\Lambda))^2 = 0 = (a_j(\Lambda))^2$. From~(\ref{twistedCAR_1}) and~(\ref{twistedCAR_2}) we obtain the twisted commutation relations,
\begin{align}
a_j\str(\Lambda) a_j(\Lambda) + \lambda_j^2 a_j(\Lambda) a_j\str(\Lambda) &= a_0(\Lambda)^2 \label{TwComm1}\\
a\str_j(\Lambda) a_k(\Lambda) &= -\lambda_j \lambda_k a_k(\Lambda) a\str_j (\Lambda) \qquad (j\neq k) \label{TwComm2}\\% \qquad\qquad\qquad a_j a_k = (-\Lambda_j)^{-1} (-\Lambda_k) a_k a_j
a_j\str(\Lambda) a_0(\Lambda) &= -\lambda_j a_0(\Lambda) a_j\str(\Lambda) \label{TwComm3}\\% \qquad\qquad\qquad a_j a_0 = (-\Lambda_j)^{-1} a_0 a_j \\
a\str_j(\Lambda) a\str_k(\Lambda) &= -\lambda_j \lambda_k^{-1} a\str_k(\Lambda) a\str_j(\Lambda) \label{TwComm4}% \qquad\qquad\qquad a_j a_k = (-\Lambda_j)^{-1} (-\Lambda_k) a_k a_j \\
\end{align}
A higher dimensional analog of the algebraic path~(\ref{sCommutation1}) is easily obtained with the following definitions for $s\in[0,s_0]$:
\begin{equation*}
{V}_{2j-1}(s) = \alpha_j(s) a_j(s) \qquad {V}_{2j}(s) = \beta_j(s)a_j(s)\str\qquad {V}_0(s) = \gamma(s) a_0(s)
\end{equation*}
where $a_j^\sharp(s) = a_j^\sharp(\Lambda(s))$ for a smooth path of vectors $\Lambda(s)$ such that $0<\lambda_j(s)<1$ for $s<s_0$. The $SO(2J+1)$ model is recovered at $s=s_0$, where $\lambda_j=1$ and the parameters $\alpha_j, \beta_j$ and $\gamma$ are given above. The commutation relations along the path can be read directly from those of the twisted CAR algebra above. Only in~(\ref{TwComm1}) do the parameters $\alpha_j,\beta_j$ and $\gamma$ appear, namely
\begin{equation}\label{addPair}
{V}_{2j}(s){V}_{2j-1}(s) + \lambda_j(s)^2 {V}_{2j-1}(s){V}_{2j}(s) = \frac{\alpha_j(s)\beta_j(s)}{\gamma^2(s)}{V}_0(s)^2\,.
\end{equation}
Using the commutation relations, the isometry condition for any $s$ reads
\begin{align*}
\idtyty = \sum_{j=1}^J\left(\abs{\beta_j}^2 - \lambda_j^2 \abs{\alpha_j}^2 \right)a_j a_j\str + \Big(\sum_{j=1}^J \abs{\alpha_j}^2 + \abs{\gamma}^2\Big) a_0^2\,,
\end{align*}
which implies 
\begin{equation} \label{sphere}
\sum_{j=1}^J \abs{\alpha_j}^2 + \abs{\gamma}^2 = 1
\end{equation}
and $\abs{\beta_j}^2 - \lambda_j^2 \abs{\alpha_j}^2 + \lambda_j^2 = 1$, or equivalently
\begin{equation*}
\abs{\beta_j}^2 = 1- \lambda_j^2 (1-\abs{\alpha_j}^2)\,.
\end{equation*}
From~(\ref{sphere}) it is natural to interpret these coefficients as angles on the sphere $\bbS^J$, with
\begin{equation*}
\gamma(s) = -\cos(s)\quad\text{and}\quad\alpha_j(s) = \frac{1}{\sqrt{J}} \sin(s)
\end{equation*}
for all $j$, and $\cos(s_0) = (2J+1)^{-1/2}$. The PVBS algebra is recovered as $s\to0$, since $\alpha_j(s)\to 0 $, while $\beta_j(s)\to 1-\lambda_j(0)^2>0$. Since the left-right symmetry has not been broken along the path, the PVBS particles come in pairs with parameters $\lambda_j(0)$ for the left edge states and $\lambda_j(0)^{-1}$ for its right edge counterpart, see~(\ref{TwComm3}) and its adjoint relation. The singularity of the representation observed in the $J=1$ case arises here too, as ${V}_{2j}(0) = 0$.

The intersection property for the spaces generated by the matrices $\{V_\alpha(s)\}_{\alpha=0}^{2J}$ follows again from Lemma~\ref{lma:IntProp}, with
\begin{equation*}
S^{(2)} = \{(00)\}\cup\{(\alpha\beta):0\leq \alpha < \beta\leq 2J\}\,.
\end{equation*}
Hence, these MPS form the ground state spaces of Hamiltonians with nearest-neighbor interaction defined by the projection onto their orthogonal complement.

In the next lemma we summarize the essential spectral properties of the transfer operator $\widehat{\bbE}_J(s)$, acting on $\caM_{2^{J}}$. We note $A\geq 0$, resp. $A>0$, whenever $A_{i,j}\geq 0$, resp. $A_{i,j}> 0$, for all $i,j$, and call these matrices nonnegative and positive, respectively.
\begin{lemma}
We consider equivalently $\widehat{\bbE}_J(s)$ or its transpose for $s>0$. $1$ is a simple eigenvalue with a positive eigenvector. Moreover, all other eigenvalues have magnitude smaller than $1$.
\end{lemma}

\begin{proof}
We drop the $s$-dependence. Once again, it is convenient to consider the $\widehat{\bbE}_J^S$ for an arbitrary subset $S\subseteq\{1,\ldots,J\}$, which is constructed from matrices $\{V_{2\alpha-1}^S,V_{2\alpha}^S\}_{\alpha\in S}$ and $V_{0}^S$, where $V_\alpha^S$ is obtained from $V_\alpha$ by keeping only the factors corresponding to indices in $S$. Let $m = \abs{S}$.

We decompose a general $2^m\times 2^m$-dimensional matrix $A = A_+\otimes \sigma^+ + A_-\otimes \sigma^- + A_P\otimes P + A_Q\otimes Q$ to obtain
\begin{align}
\widehat{\bbE}_J^S(A) &= -\lambda_j \widehat{\bbE}_J^{S^{j}}(A_+)\otimes \sigma^+ -\lambda_j \widehat{\bbE}_J^{S^{j}} (A_-)\otimes \sigma^- \nonumber \\ 
& \quad \left[\widehat{\bbE}_J^{S^{j}} (A_P) + \abs{\alpha_j}^2 (V_0^{S^{j}})\str A_Q V_0^{S^{j}}\right]\otimes P + \left[\abs{\beta_j}^2 (V_0^{S^{j}})\str A_P V_0^{S^{j}} + \abs{\lambda_j}^2\widehat{\bbE}_J^{S^{j}}(A_Q)\right]\otimes Q\,.
 \label{J_Matrix}
\end{align}
First, the eigenvalue of $\widehat{\bbE}_J^S$ corresponding to an `off-diagonal' eigenvector must be smaller than any eigenvalue of $\widehat{\bbE}_J^{S^{j}}$ by a factor $-\lambda_j$. But $\widehat{\bbE}_J^{S^{j}}$ is also a principal submatrix of the diagonal part, which is nonnegative. Hence, its spectral radius must itself be smaller and we can restrict our attention to $\widehat{\bbE}_J^S\upharpoonright_\mathrm{diagonal}$. Its matrix representation has the following block form,
\begin{equation*}
M^S = \begin{pmatrix}
M^{S^j} & \abs{\alpha_j}^2 D^{S^j} \\ \abs{\beta_j}^2 D^{S^j} & \abs{\lambda_j}^2 M^{S^j}
\end{pmatrix}
\end{equation*}
where $D^{S^j}$ is the diagonal matrix containing all possible products $\prod_{i_1,\ldots i_k\neq j}\lambda_{i_1}\cdots\lambda_{i_k}$ for $0\leq k\leq m-1$. If $a,b$ index the blocks, it is immediate from
\begin{equation*}
\left[(M^S)^m\right]_{a,b}= \sum_{c_1,\ldots,c_{m-1}=1}^2 M^S_{a,c_1}\cdots M^S_{c_{m-1},b}
\end{equation*}
that for any $a,b$, the submatrix $M^{S^j}$ appears at least $m-1$ times in at least one term of the sum. Since the other factors are diagonal with positive matrix elements, this implies that $(M^S)^m>0$ if $(M^{S^j})^{m-1}>0$. Since the matrix depleted to just one particle reads
\begin{equation*}
M^{\{1\}} = \begin{pmatrix}
\abs{\gamma}^2 & \abs{\alpha_1}^2 \\ \abs{\beta_1}^2 & \abs{\lambda_1}^2 \abs{\gamma}^2 
\end{pmatrix} >0 \,,
\end{equation*}
we conclude by recursion that $(M^S)^m>0$. But this implies that the nonnegative matrix $M^S$ is irreducible, so that by Perron-Frobenius' theorem, its spectral radius is a simple eigenvalue with a positive eigenvector, and all other eigenvalues have a smaller magnitude.

By construction, the identity matrix is an eigenvector of the full $\widehat{\bbE}_J$ for the eigenvalue $1$, which implies that $1$ is the spectral radius of $\widehat{\bbE}_J$. This concludes the proof of the lemma for $\widehat{\bbE}_J$. Its transpose shares its spectrum, and by nonnegativity and irreducibility, it must also have a positive eigenvector.
\end{proof}

Clearly, the proof collapses at $s=0$ as $\alpha_j(0)=0$. In fact, the matrix $M^S$ becomes lower triangular and therefore reducible, and the eigenvectors fail to be positive. We recover the singularity of the path of the previous section: The positive diagonal left eigenvector $\rho_J$ of $\widehat{\bbE}_J$, normalized to represent a density matrix, corresponds to a faithful state for $s>0$, but not anymore at $s=0$. Once again, this is necessary as the minimal representation of the PVBS model is one dimensional.

The commutation relations allow for a simple description of the ground state space along the path of Hamiltonians. For each pair of particle indexed by $j$, the matrices $V_{2j}$ and $V_{2j-1}$ must alternate, separated by an arbitrary number of any other $V_\alpha$. If the chain is long enough, this produces $4^J$ different possibilities, depending on whether the first and the last matrix of each of the $J$ types has en even or an odd index. Since the initial space of the matrix product map is $\caM_{2^J}$, we have a complete description of the ground state spaces along the path.

Part~(i) of Lemma~\ref{lma:Commutation} holds for arbitrary $J$, as it requires only the spectral picture of $\widehat{\bbE}_J$, the fact that $\idtyty$ is the eigenvector for the largest eigenvalue $1$, and the faithfulness of $\rho_J$. In other words, the path is gapped for all $s>0$. The proof of part~(ii) would rely on the complete diagonalization of the transfer operator, which will not be carried out here. However, the key element is the `cost' of adding a pair above the PVBS ground state. From~(\ref{addPair}), it is precisely $\alpha_j\beta_j / \gamma_j^2$ which is proportional to $\tan(s) / (\sqrt{J}\cos(s))$ for all $j$, indicating that the perturbative argument can be carried out in the general case. This would imply that the gap does not close as $s\to 0$, and the generalization of Theorem~\ref{thm:AKLT_phase}: For any $J\in\bbN$, the $SO(2J+1)$-invariant model belongs to the PVBS phase with $n_L = n_R = J$.

The paths constructed here and the necessary singularity encountered at the PVBS point raise the more general question of the continuity of the spectral gap of an interaction exposing a continuous family of matrix product states, in the spirit of~\cite{BravyiHastings} or~\cite{Stability}. The fundamental difficulty lies in the non invertibility of the eigenvectors of $\widehat{\bbE}$ at the PVBS point, and therefore in the collapse of standard estimates involving $\rho^{-1}$ and $e^{-1}$. Such `stability' results are notoriously rare (see~\cite{Ya06} for a concrete example) but we do believe that continuity holds more generally than in the simple case presented here. This is subject to current investigation. 

%%%%%%%%%%%%%%%%%%%%%%%%%%%%%%%%%%%%%%%%%%%%%%%%%%%%%%%%%%%%%%%%%%%%%

\section*{acknowledgments} The authors gratefully acknowledge the kind hospitality of the Erwin Schr\"odinger International Institute for Mathematical Physics (ESI) in Vienna, Austria, where part of the work reported here was carried out. S.B. wishes to thank Y. Ogata for her hospitality at the Graduate School of Mathematical Sciences, University of Tokyo, Japan, and for stimulating discussions. Both authors thank Amanda Young for useful remarks. This work was supported in part by the National Science Foundation: S.B. under Grant \#DMS-0757581 and B.N. under grant \#DMS-1009502

\bibliography{RefGap}
\bibliographystyle{unsrt}

\end{document}